\documentclass[pdflatex,sn-mathphys,iicol]{sn-jnl}
\usepackage{xspace}
\usepackage{url}
\usepackage{amssymb}
\usepackage{caption}

\jyear{2023}%

\theoremstyle{thmstyleone}%
\newtheorem{theorem}{Theorem}
\newtheorem{lemma}[theorem]{Lemma}

\theoremstyle{thmstyletwo}%

\theoremstyle{thmstylethree}%
\newtheorem{definition}{Definition}%

\begin{document}

\title[ ]{Scheduling Periodic Messages on a Shared Link without Buffering}


\author*[1,2]{\fnm{Ma\"el} \sur{Guiraud}}\email{mguiraud@cesi.fr}
\equalcont{These authors contributed equally to this work.}
\author*[2]{\fnm{Yann} \sur{Strozecki}}\email{yann.strozecki@uvsq.fr}
\equalcont{These authors contributed equally to this work.}

\affil*[1]{\orgdiv{LINEACT}, \orgname{CESI Nanterre}, \orgaddress{\street{93 boulevard de la Seine }, \city{Nanterre}, \postcode{BP 602 Cedex 92006}, \country{France}}}

\affil[2]{\orgdiv{DAVID Laboratory}, \orgname{Université de Versailles-Saint-Quentin}, \orgaddress{\street{45 Avenue des Etats-Unis}, \city{Versailles}, \postcode{78000},  \country{France}}}

\newcommand{\Fo}{\textsf{FO}} 
\newcommand{\NP}{\textbf{NP}} 
\newcommand{\FPT}{\textbf{FPT}} 
\newcommand{\Fmo}{\textsf{FMO}}
\newcommand\pma{\textsc{pma}\xspace}
\newcommand\firstfit{\texttt{First Fit}\xspace}
\newcommand\compactpair{\texttt{Compact Pairs}\xspace}
\newcommand\metaoffset{\texttt{Meta Offset}\xspace}
\newcommand\greedyuniform{\texttt{Greedy Uniform}\xspace}
\newcommand\swapandmove{\texttt{Swap and Move}\xspace}
\newcommand\compactfit{\texttt{Compact Fit}\xspace}
\newcommand\greedypotential{\texttt{Greedy Potential}\xspace}
\newcommand{\todo}[1]{{\color{red} TODO: {#1}}}








\abstract{
Cloud-RAN, a novel architecture for modern mobile networks, relocates processing units from antenna to distant data centers. This shift introduces the challenge of ensuring low latency for the periodic messages exchanged between antennas and their respective processing units. In this study, we tackle the problem of devising an efficient periodic message assignment scheme under the constraints of fixed message size and period \emph{without contention nor buffering}.

We address this problem by modeling it on a common network topology, wherein contention arises from a single shared link servicing multiple antennas. While reminiscent of coupled-task scheduling, the introduction of periodicity adds a unique dimension to the problem. We study how the problem behaves with regard to the \emph{load} of the shared link, and we focus on proving that, for load as high as possible, a solution \emph{always} exists and it can be found in polynomial time.

The main contributions of this article are two polynomial-time algorithms,  which find a solution for messages of any size and load at most $2/5$ or for messages of size one and load at most $\phi - 1$, the golden ratio conjugate. We also prove that a randomized greedy algorithm finds a solution on almost all instances with high probability, shedding light on the effectiveness of greedy algorithms in practical applications.
}


\maketitle

\section{Introduction}

The Radio Access Networks (RAN) architecture is the part of the mobile phone network which communicates with mobile handsets. It is composed of base stations managing radio emissions and multiple computations, connected to the core network~\cite{hossain2019recent}. 
An objective of 5G+ is to split a base station into two parts: The Remote Radio Head (RRH), in charge of the radio emissions, and the Baseband Unit (BBU\footnote{Others terminologies exist in the literature. The results of this work are fully compatible with any variation of the C-RAN architecture.}) in charge of the computations. In the Cloud RAN architecture (C-RAN), to reduce maintenance and energy consumption costs~\cite{gavrilovska2020cloud,mobile2011c,checko2014cloud}, the BBUs are gathered in one or several data centers and are connected to the RRHs via the \emph{fronthaul} network.
 The main challenge of C-RAN is to reach a latency compatible with transport protocols~\cite{ieeep802}, to support functions like  HARQ (Hybrid Automatic Repeat reQuest) in only $3$ms~\cite{bouguen2012lte}. The latency is measured between the sending of a message by an RRH and the reception of the answer, computed by a BBU in the cloud. In addition to the latency constraint, the specificity of C-RAN is the periodicity of the data transfer in the fronthaul network between RRHs and BBUs: messages need to be emitted and received each millisecond~\cite{dogra2020survey,3gpp5g,romano2019imt}.

Our approach is based on the URLLC (Ultra-Reliable Low-Latency Communication) context~\cite{siddiqui2023urllc}, aiming to ensure minimal packet loss and low latency in networks. In addition to minimum latency and $0\%$ packet loss, we guarantee no jitter in our streams. To achieve this, it is essential to have guarantees that the network can deliver messages on specific dates, using a central controller that activates output ports on given slots. These guarantees are based on the solutions proposed by the TSN (Time Sensitive Networking) group~\cite{ieee802}, in particular the 802.1 Qcc standard for central network control and the 802.1 Qbv standard for individual flow management. This desire to control the network using a software controller is the very essence of Software Defined Networking (SDN)~\cite{waseem2022software}.

The latest generation of Radio Access Networks (RANs), known as O-RAN~\cite{polese2023understanding}, merges the aspiration to relocate Baseband Units (BBUs) from Cloud RAN Remote Radio Heads (RRHs) with the necessity to manage networks through an additional software layer. A prototype already exists, substantiating the practical feasibility of our approach~\cite{guiraud2022experimental,leclerc2016transmission,marce2018Coordinated}.

Our aim is to operate a C-RAN on a low-cost shared switched network: several (tens of) antennas share a high-speed link to send their periodic messages to one (or several) data center. This shared link is the only contention point for a message going to the data center and it is also the contention point for the answer sent back by the data center to the antenna. This model with two contention points (one for the message and one for its answer) also captures other periodic systems such as processors communicating over a bus or sensors doing periodic radio transmissions on the same frequency.

We address the following question: \emph{is it possible to schedule periodic messages on a shared link without using buffers}? Eliminating this source of latency leaves us with more time budget for latency due to the physical length of the routes in the network, and thus allows for wider deployment areas. Our proposed solution is to compute beforehand a \emph{periodic and deterministic} sending scheme, which completely avoids contention.  While a sending scheme without buffering is simpler and less expensive to implement in C-RAN networks than one with buffering, it may not exist and we have previously investigated the case with a buffer~\cite{barth2018deterministic}. 

The algorithmic problem studied, called \emph{Periodic Message Assignment} or \pma, is as follows: Given a period, a message size, and, for each message a delay between the two contention points, set a departure time in the period for each message so that they go through both contention points without collision. All values are integers, a unit of time corresponds to the time to transmit a minimal quantity of data over a link. 

The load of a network with a link shared by all messages is the load of this link, defined as the ratio of the bandwidth used to the bandwidth available. In our context, it is the time used by messages on the link in a period, divided by the period. When the load is small, it is easier to design sending scheme without collision, as already noted in~\cite{barth2018deterministic,guiraud2021deterministic}.
The aim of this article is to \textbf{determine the largest load} under which it is \textbf{always possible to find a sending scheme} and to give \emph{polynomial time} algorithms to produce it. 
Knowing this load enables the link to be correctly sized. Increasing this load, we can have more antennas sharing the same link, which reduces the cost of deployment of the fronthaul network. Moreover, a method to mix the traffic of the antennas with random sources of traffic requires finding sending schemes for a load as high as possible~\cite{guiraud2021deterministic}. 

A pre-print version of this paper is available at \cite{DBLP:journals/corr/abs-2002-07606}.

\paragraph*{Related Works}

The model we study was recently introduced in~\cite{barth2018deterministic,guiraud2021deterministic} to find sending schedules for C-RAN messages, with \emph{buffering allowed}. This problem has also been studied for a cycle topology instead of a shared link~\cite{Guir1905:Deterministic}. In these articles, the main results are heuristics, using classical scheduling algorithms as subroutines, and fixed-parameter tractable algorithms which find a sending scheme with \emph{minimal latency}. The problem of finding sending schemes with \emph{no additional latency}, which is the subject of this article under the name \pma, is introduced for the first time and briefly studied in~\cite{bartharxiv2018deterministic}.

If we ignore periodicity in the Periodic Message Assignment problem, it can be compared to several classical scheduling problems. The Periodic Message Assignment problem is similar to a coupled-task scheduling problem~\cite{shapiro1980scheduling,khatami2020coupled,chen2021scheduling} in a two flow-shop environment. Scheduling coupled-tasks is $\NP$-complete in most settings~\cite{orman1997complexity}, while for identical jobs (all delays are the same) it is in polynomial time~\cite{baptiste2010note}. Applications cited in these works are often related to radar transmission, which is very similar to our problem since it involves sending and receiving a given quantity of data with a fixed interval of time betwen sending and receiving.

Alternatively, \pma can be interpreted as a two-machine flow shop scheduling problem~\cite{khatami2023flow,zhao2020two,johnson1954optimal,yu2004minimizing} with an exact delay between tasks. The problem with an exact delay has been investigated and found to be $\NP$-hard except when all delays are equal~\cite{leung2007scheduling}. The periodicity adds more constraints since the sending pattern for a single period must be repeated without creating collision at contention points. When considering a two-machine flow shop problem, the aim is usually to minimize the makespan, the schedule length, or the sum of job completion times. In our periodic variant, these quantities are irrelevant and we look for any feasible periodic schedule without buffering, that is respecting the exact delay between the two coupled tasks.

 To our knowledge, periodic scheduling problems studied in the literature are quite different from the problem studied in this article. The problem of scheduling periodic tasks dates back to the 70s~\cite{liu1973scheduling}, with an emphasis on the maximal load as in this article. In~\cite{liu1973scheduling} only a single contention point is modeled and the scheduling is preemptive, which makes the proposed algorithms irrelevant to our problem.

 Variations on the problem of minimizing latency of periodic messages in networks have been considered and practically solved, using mixed-integer programming~\cite{nayak2017incremental,steiner2018traffic} or an SMT solver~\cite{dos2019tsnsched}, but without theoretical guarantees on the quality of the produced solutions nor on the computation time. Recent work that aims to schedule periodic flows using TSN technologies employs a model very similar to ours in terms of architecture and the objective of eliminating jitter~\cite{9472838}. However, the studied model does not account for round-trip packets, and the authors provide only a greedy algorithm. Typical applications cited in these works (out of C-RAN) are sensor networks communicating periodically inside cars or planes, or logistic problems in production lines, which can also be captured in our model.

 In another line of work~\cite{korst1991periodic,hanen1993cyclic}, the aim is to minimize the number of processors on which the periodic tasks are scheduled, while our problem corresponds to two fixed and different processors. In cyclic scheduling~\cite{levner2010complexity}, the aim is to minimize the period of a schedule to maximize the throughput, while our period is fixed. 

 The train timetabling problem~\cite{lusby2011railway} and in particular, the periodic event scheduling problem~\cite{serafini1989mathematical} or cyclic train timetabling~\cite{zhang2019solving} are generalizations of our problem since they take the period as input and can express the fact that two trains (like two messages) should not cross. However, they are much more general: the trains can vary in size and speed, the network can be more complex than consecutive single tracks and there are precedence constraints. Hence, the numerous variants of train scheduling problems are very hard to solve. Therefore, some delay is allowed in different parts of the network to make these problems solvable and most of the research done~\cite{lusby2011railway} is devising practical algorithms using branch and bound, mixed-integer programming, genetic algorithms, etc. 

The approach of employing dynamic deterministic flow calculation is emerging. While~\cite{9234005} presents a genetic algorithm based on a job-shop scheduling model that deviates significantly from our model,~\cite{gartner2023fast} uses incremental approaches with a concept similar to the \swapandmove algorithm we introduce: minimizing the impact of scheduling choices on future packets. The authors do not try to optimize overall latency but focus solely on computation time, using topologies and flow modelization different from ours.

In this work, we do not compare our algorithms with classical methods used to solve similar problems such as mixed-integer linear programming, simulated annealing or genetic algorithm. Indeed, our objective is not design the best heuristic to solve \pma, but rather to understand for which load this problem has always a solution. While it is possible to analyze greedy algorithms to prove such a result, it seems hard to do the same for more complex heuristics. Besides, \pma is a constraint satisfaction problem and not an optimization problem, making some classical methods less relevant. The periodicity makes the constraints hard to state in a linear fashion, which makes the problem harder to cast as a linear programming problem (though not impossible, see~\cite{bartharxiv2018deterministic}). 

\paragraph*{Contributions}

Our primary objective is to establish that for relatively small loads, a scheduling for the \pma problem always exists and can be determined in polynomial time. 
This extends previous work of the authors~\cite{bartharxiv2018deterministic}, where it was proven that the greedy algorithm, denoted as \metaoffset, guarantees a scheduling when the load is less than $1/3$.
Our first contribution is the design and analysis of the \texttt{Compact k-tuples} algorithm. This sophisticated greedy algorithm schedules carefully chosen tuples of messages simultaneously.
The algorithm operates in polynomial time and ensures a scheduling for loads up to $2/5$.

Our second contribution is the design and analysis of the \swapandmove algorithm, in the context of messages of size $1$. This algorithm does local improvements to improve the packing of messages. It operates in polynomial time and always finds a scheduling when the load is less than $(\sqrt{5}-1)/2 \approx 0,618$. This is a significant improvement over existing greedy algorithms, which only guarantee a scheduling for loads less than $1/2$.

Our third contribution is a set of reductions, which show that solving \pma whith message of size one is sufficient to address similar problems with message of any size, buffering and a general fronthaul network, as opposed to a single link. This expands the applicability of our results, particularly in the context of C-RAN.

Our final contribution is an experimental study, comparing the quality and runtime performance of various algorithms in solving the \pma problem on random instances. Notably, our results show that two algorithms, \compactfit (a simplified version of \texttt{Compact k-tuples}) and \swapandmove, consistently outperform other algorithms. We have made the source code for these algorithms available on github\footnote{\url{https://github.com/Mael-Guiraud/GuiraudStrozecki2023Scheduling}}. Moreover, our experiments reveal that schedules can be found for significantly higher loads than guaranteed by our theoretical proofs. To explain this phenomenon, we prove that \greedyuniform, the simplest randomized greedy algorithm, is capable of finding a scheduling for nearly all inputs with loads strictly less than 1.  In fact, when compared to an ad-hoc exhaustive search algorithm proposed in~\cite{bartharxiv2018deterministic}, our best algorithm finds almost as many solutions.

For a summary of the performance of the algorithms presented in this paper, please refer to Table~\ref{tab1}. It is worth noting that all algorithms discussed in this article are novel, with the exception of \metaoffset.

\begin{table*}[h]
\begin{minipage}{400px}
\begin{center}
\caption{Summary of the main results of this paper. We give the maximum load for which an algorithm always finds a solution. Integer $n$ is the number of scheduled messages. For experimental results on random instances, see Sec.\ref{sec:perf}.}\label{tab1}%
\begin{tabular}{@{}|l|l|p{52mm}|l|@{}}
\hline
Algorithm & Message size & Maximum Load & Complexity\\
\hline
\firstfit    & 1 & 1/2  & $O(n^2)$ \\
\swapandmove & 1 & $(\sqrt{5}-1)/2 \approx 0,618$ & $O(n^3)$\\
\greedyuniform & 1 & $\forall \varepsilon > 0, \,1 - \varepsilon$ with high probability, \newline 
for large random instances  & $O(n^2)$  \\
\compactpair & 2 & 4/9  & $O(n^2)$ \\
\firstfit    & $\geq 1$ & 1/3  & $O(n^2)$ \\
\metaoffset  & $\geq 1$ & 1/3  & $O(n^2)$ \\
\compactpair & $\geq 1$ & 3/8  & $O(n^2)$ \\
\texttt{Compact 8-tuples} & $\geq 1$ & $\geq 2/5$ & $O(n^2)$\\
\hline
\end{tabular}
\end{center}
\end{minipage}
\end{table*}

\paragraph*{Organization of the Paper}

In Sec.~\ref{sec:model}, we explain how the network is modeled and we introduce the problem \pma. In Sec.~\ref{sec:large}, we present several greedy algorithms and prove they always find a solution to \pma for moderate loads. These algorithms rely on schemes to build compact enough solutions, to bound measures of the time wasted when scheduling messages. 

In Sec.~\ref{sec:small} we present deterministic and probabilistic algorithms for messages of size $1$, which work for much higher loads than the algorithms designed for arbitrary messages. The deterministic algorithm is not greedy, contrarily to algorithms of Sec.~\ref{sec:large}, since it uses a swap mechanism that moves already scheduled messages. 

In Sec.~\ref{sec:gen}, we prove that hypotheses on the period and the message size can be relaxed. We also show that \pma captures periodic message scheduling in networks with many contention points, as long as the routing is coherent.
Finally, we present the performance of all algorithms on random inputs, both on large messages in Sec.~\ref{sec:perf_large} and small messages in Sec.~\ref{sec:perf_small}.
It shows that the algorithms of this article work for random instances with a much higher load than what has been proved in the worst case.

\section{Modeling a C-RAN Network}\label{sec:model}
\begin{center}
\begin{figure}

\centering
\includegraphics[scale=0.22]{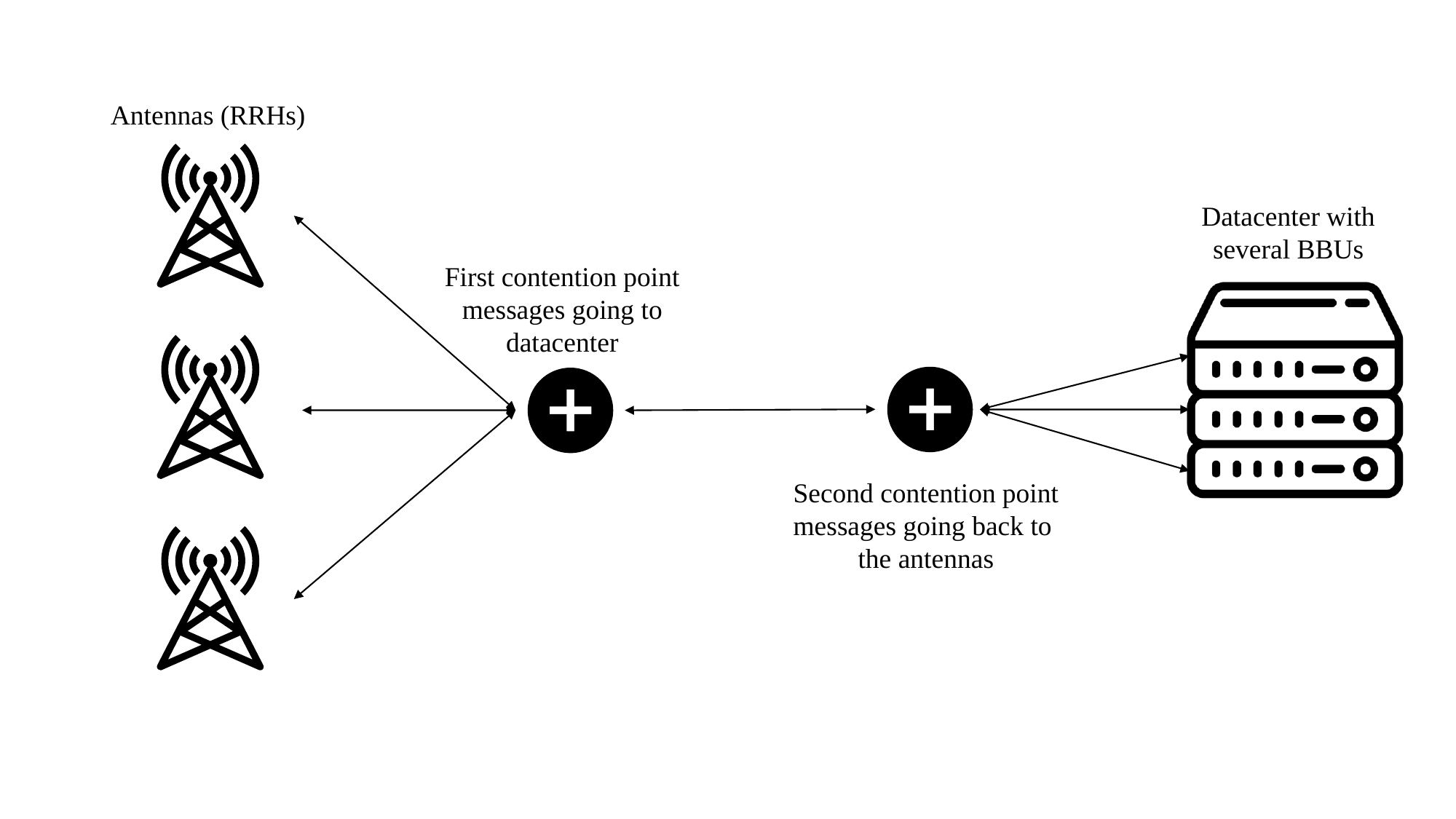} 

\scalebox{0.45}{
\includegraphics[]{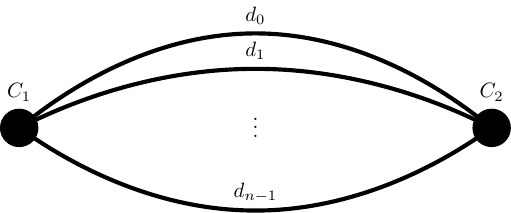}
}

\caption{C-RAN network with a single shared link modeled by two contention points and delays}
\label{fig:model}
\end{figure}
\end{center}

In this article, we model a simple network in which periodic messages flow through a single bidirectional link. Messages using the link in two different directions do not interact, since the link we model is full-duplex.
 Each RRH sends messages through the link to its BBU and two messages cannot go at once in the link: this is the first contention point, represented in Fig.~\ref{fig:model}. Upon receiving a message, a BBU sends an answer back to its RRH, which goes through the link in the other direction: this is the second contention point represented in Fig.~\ref{fig:model}. Since the answer must be sent back as soon as a message arrives, we see this process as a single message going from an RRH to its BBU and back to the RRH, while traversing two contention points.

The \textbf{size} is an integer representing the time needed to send a message through a contention point of the network, here the beginning of the link shared by all antennas. In the C-RAN context we consider, \emph{all messages are of the same kind}, hence they are all of the same size denoted by $\tau$. Indeed, the message sent by the RRH to the BBU is the raw electromagnetic signal it has captured during a millisecond. The message sent by the BBU to the RRH is the electromagnetic signal it must emit during a millisecond.

 We denote by $n$ the number of messages, which are numbered from $0$ to $n-1$. A message $i$ is characterized by its integer \textbf{delay} $d_i$: when message number $i$ arrives at the beginning of the shared link (first contention point) at time $t$, it returns to the other end of the link on its way back (second contention point) at time $t + d_i$. The delay represents the transmission time from the RRH to the BBU plus the processing time in the BBU. 

  The model and problem can easily be generalized to any topology, that is any directed acyclic multigraph with any number of contention points, see~\cite{bartharxiv2018deterministic}. We choose here to focus on a realistic network with a single shared link, which is simple enough to obtain theoretical results. It turns out that algorithms solving the problem for a single shared link can be used on networks with many contention points, as long as the routing is coherent, see Sec.~\ref{sec:coherent}. 

The time is \emph{discretized} and the process we consider is \emph{periodic} of fixed integer \textbf{period $P$}. We use the notation $[P]$ for the set $\{0,\dots,P-1\}$. A message is emitted an infinite number of times periodically, hence it is enough to consider any interval of $P$ units of time to completely represent the state of our system by giving the times, in this interval, at which each message goes through the two contention points. We call the representation of an interval of $P$ units of time in the first contention point the \textbf{first period} and the \textbf{second period} for the second contention point. Because the system is of period $P$, we may always assume that $d_i$, the delay of message $i$, is in $[P]$. 

An \textbf{offset} of a message is a choice of time at which it arrives at the first contention point (i.e. in the first period). Let us consider a message $i$ of offset $o_i$, it uses the interval of time $[i]_1 = \{ (o_i + t) \mod P \mid 0 \leq t < \tau \}$ in the first period and $[i]_2 = \{ (d_i + o_i + t) \mod P \mid 0 \leq t < \tau \}$ in the second period, as illustrated in Figure~\ref{fig:offset}. Two messages $i$ and $j$ \textbf{collide} if either $[i]_1 \cap [j]_1 \neq \emptyset $ or $[i]_2 \cap [j]_2 \neq \emptyset $. If $t \in [i]_1$ (resp. $t \in [i]_2$), we say that message $i$ uses time $t$ in the first period (resp. in the second period). Figure~\ref{fig:collide} shows an example of a collision.
 \begin{center}
\begin{figure}

\centering
\includegraphics[scale=0.28]{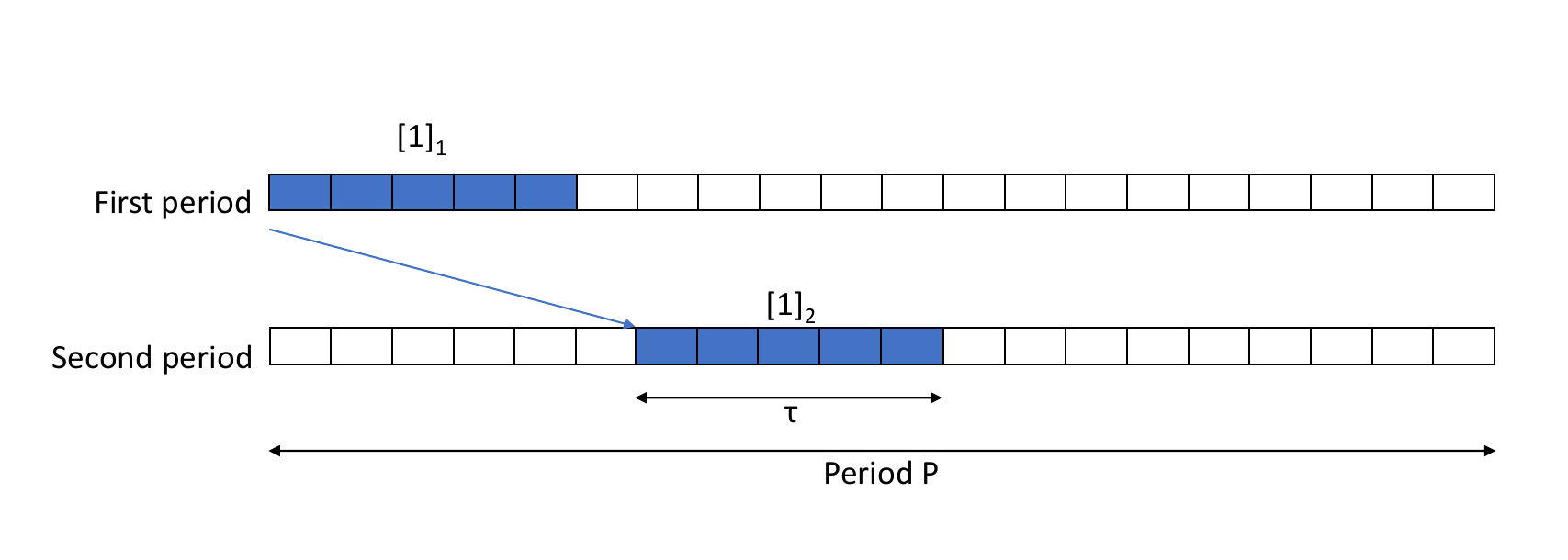} 

\caption{A message with offset $0$ and delay $6$.}
\label{fig:offset}
\end{figure}
\end{center}
 \begin{center}
\begin{figure}

\centering
\includegraphics[scale=0.28]{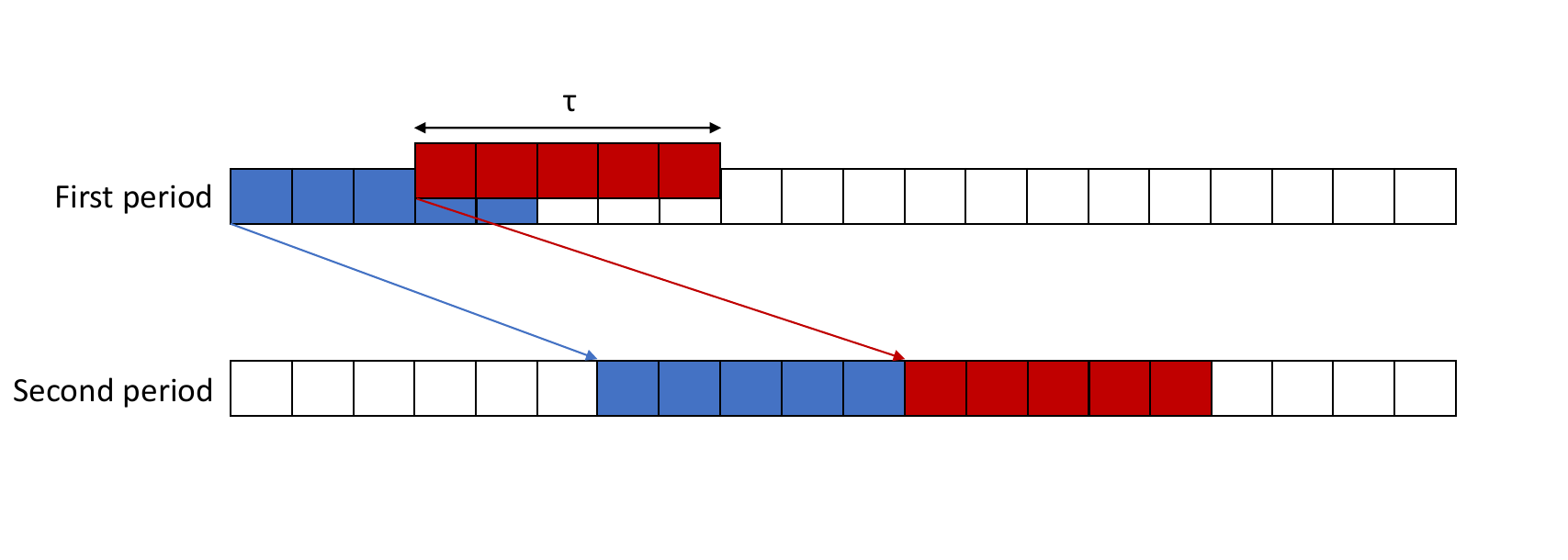} 

\caption{The messages $1$ and $2$ with offset $0$ and $3$ collide because $[1]_1 \cap [2]_1 = \{3,4\} \neq \emptyset$.}
\label{fig:collide}
\end{figure}
\end{center}
Our objective is to send all messages so that there is no collision in the shared link.
To do that, we can choose the offset of each message. An \textbf{assignment} $A$ is a function from $[n]$ to $[P]$. The value $A(i)$ is the offset of the message $i$. We say that an assignment is \textbf{valid} if \emph{no pair of messages collide}, as shown in Fig.~\ref{fig:assignment}.

Let \textbf{Periodic Message Assignment} or \pma be the following problem: given $n$ messages of delays $d_0,\dots,d_{n-1}$, a period $P$ and a size $\tau$, find a valid assignment or decide there is none. When a valid assignment is found, we say the problem is solved \textbf{positively}.

\begin{figure*}
\begin{center}
\includegraphics[scale=0.3]{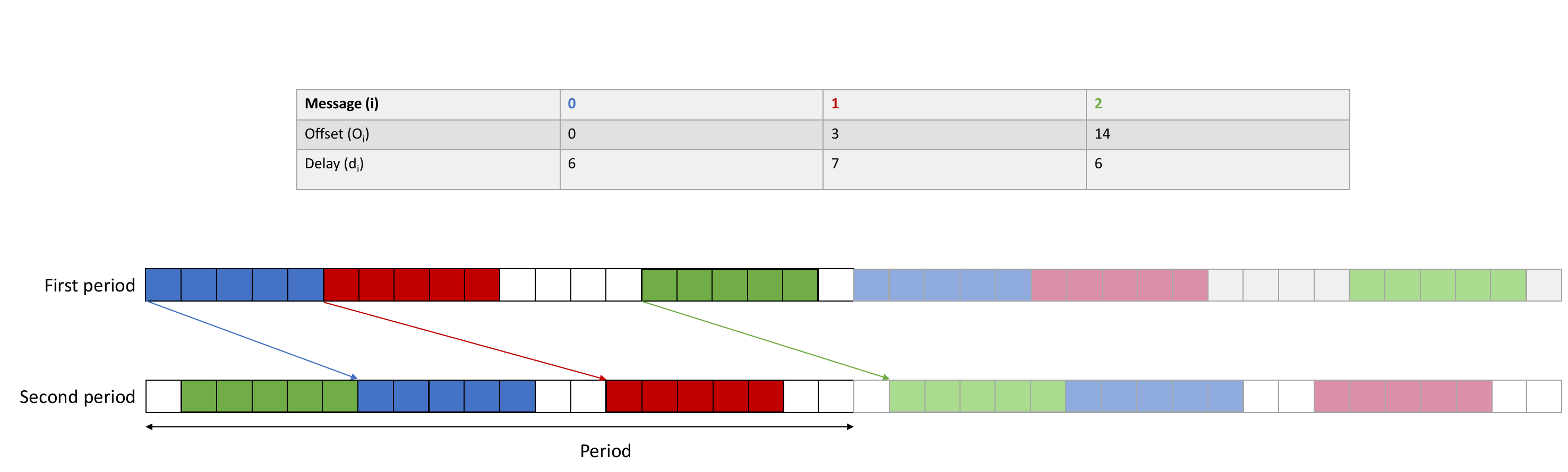}
\end{center}
\caption{An instance of \pma with $3$ messages, $P= 20$, $\tau = 5$, and one assignment}
\label{fig:assignment}
\end{figure*}

It is yet unknown whether \pma is $\NP$-hard. However, it has been proven that, when parameterized by $n$ the number of messages, the problem is \FPT~\cite{barth2018deterministic}. A slight generalization of \pma presented in Sec.~\ref{sec:coherent}, with more contention points, but each message only going through two of them as in \pma, is \NP-hard~\cite{barth2018deterministic}. If the shared link is not full-duplex, that is, there is a single contention point and each message goes through it twice, it is also \NP-hard, since we can encode a similar non-periodic problem~\cite{orman1997complexity}. Hence, we conjecture that \pma is \NP-hard.
	
Because we are interested in \pma when it can always be solved positively, we study it when the load of the system is small enough. The \textbf{load} is defined as the number of units of time used in a period by all messages divided by the period that is $n\tau /P$. There cannot be an assignment when the load is larger than one; we prove in this article that, for moderate loads, there is \emph{always} an assignment and that it can be found by polynomial time algorithms. This kind of result is helpful when solving the following optimization version of \pma: given a set of messages, find the largest subset which admits an assignment. A weighted version, where the messages have different values can also be considered. An optimal solution to the optimization problem is a set of messages corresponding to a load of at most $1$. Assume we have an algorithm that always finds an assignment for an instance of load $\lambda$. Then, such an algorithm finds an assignment for any subset of load $\lambda$ and is an approximation algorithm for the optimization problem with approximation ratio $\lambda$.

\section{Greedy Algorithms for Arbitrary Messages} \label{sec:large}

In this section, we study the case of arbitrary values for $\tau$. When modeling a C-RAN network, we choose the time granularity, and we could set it so $\tau = 1$ for simplicity. However, the length of a link and thus the delay of a message is typically of the same magnitude as $\tau$, therefore setting $\tau = 1$ is a too coarse granularity to faithfully model the network.

A \textbf{partial assignment} $A$ is a function defined from a subset $S$ of $[n]$ to $[P]$.
The cardinality of $S$ is the \textbf{size} of partial assignment $A$. A message in $S$ is \textbf{scheduled} (by $A$), and a message not in $S$ is \textbf{unscheduled}. We only build valid partial assignments: no pair of messages of $S$ collide. If $A$ has domain $S$, and $i \notin S$, we define the extension of $A$ to the message $i$ by the offset $o$, denoted by $A[i \rightarrow o]$, as $A$ on $S$ and $A[i \rightarrow o](i) = o$.

All presented algorithms build an assignment incrementally, by growing the size of a valid partial assignment. Moreover, algorithms in this section are \emph{greedy}: Once an offset is chosen for a message, it is never changed. In the rest of the paper, we sometimes compare the relative position of messages to detect collisions, but one should remember that the time is periodic and these are relative positions on a circle. 
In some remarks and computations, we may omit to write \emph{mod P}, in order to not overburden the presentation.

 In this section, we first present two simple greedy algorithms:
 \firstfit which produces compact assignments and \metaoffset which 
 relies on the absence of collision in the first period. We then propose \texttt{Compact k-tuples}, a family of algorithms that combine ideas from the two previous algorithms and work by scheduling tuples of messages.

\subsection{First Fit}

Consider some partial assignment $A$, in the first period, the message $i$ uses times in $[i]_1 = [A(i), A(i) + \tau -1 \mod P]$. If a message $j$ is scheduled by $A$, with $A(j) < A(i)$, then the last time it uses in the first period is $A(j)+\tau-1$ and it should be less than $A(i)$, which implies that $A(j) \leq A(i) - \tau$. Symmetrically, if $A(j) > A(i)$, to avoid collision between messages $j$ and $i$, we have $A(j) \geq A(i) + \tau$. Hence, message $i$ forbids the interval $[A(i) - \tau +1 \mod P, A(i) + \tau -1 \mod P]$ as offsets for messages still not scheduled, because of its use of time in the first period. The same reasoning shows that $2\tau -1$ offsets are also forbidden because of the times used in the second period. Hence, if $\lvert S\rvert $ messages are already scheduled, then at most $\lvert S\rvert(4\tau -2)$ offsets are forbidden for an unscheduled message. The real number of forbidden offsets may be smaller since the same offset can be forbidden both because of a message on the first and on the second period.

To formalize the idea of the previous paragraph, we introduce $\Fo(A)$, which is the number of \emph{offsets forbidden by $A$}. Let $A$ be a partial assignment defined over $S$ and $i\notin S$, $\Fo(A)$ is the maximum over all values of $d_i \in [P]$ of $\lvert\left\{ o \in [P] \mid A[i \rightarrow o] \text{ has a collision}\right\}\rvert$. In the previous paragraph, we have proved that $\Fo(A)$ is at most $(4 \tau -2)\lvert S\rvert$, as shown if Figure~\ref{fig:FO}. 

\begin{figure}
\begin{center}
\includegraphics[scale=0.39]{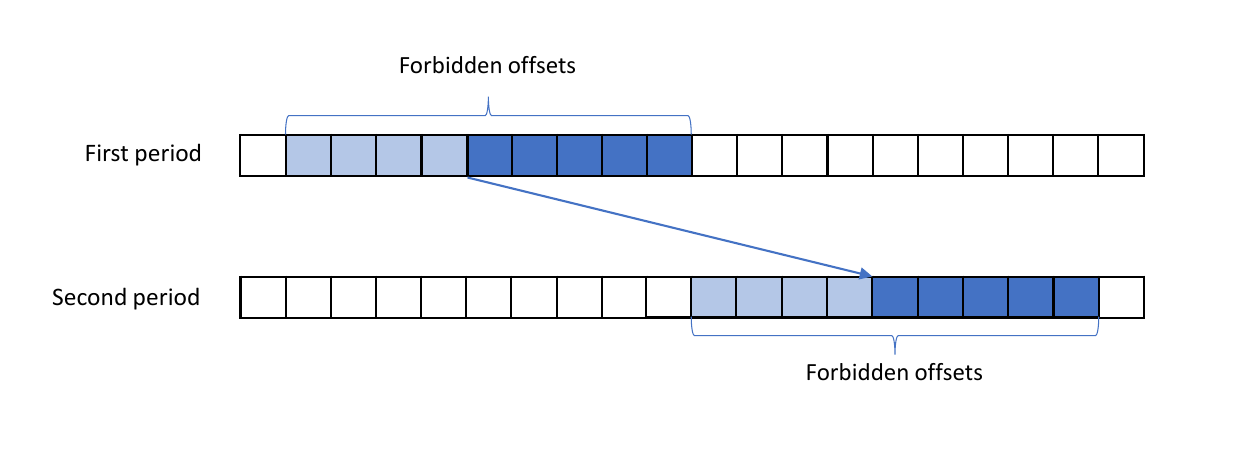}
\end{center}
\caption{An example with $\lvert S\rvert = 1$ and an unscheduled message of delay $0$, where $FO(A) = (4 \tau -2)\lvert S\rvert$. }
\label{fig:FO}
\end{figure}

Let \firstfit be the following algorithm:  for each unscheduled message (in the order they are given), try all offsets from $0$ to $P-1$ until one does not create a collision with the current assignment and use it to extend the assignment. 

When $\Fo(A) < P$, whatever the delay of the message we want to extend $A$ with, there is an offset to do so. Since $\Fo(A) \leq (4 \tau -2)\lvert S\rvert$ and $\lvert S\rvert < n$, \firstfit (or any greedy algorithm) always succeeds when $(4 \tau -2)n \leq P$, that is when the load $ n\tau /P$ is at most $1/4$.
It turns out that \firstfit always creates compact assignments (as defined in~\cite{bartharxiv2018deterministic}), that is a message is always next to another one in one of the two periods. Hence, we can prove a better bound on $\Fo(A)$, when $A$ is built by \firstfit, as stated in the following theorem.

\begin{theorem}
\firstfit solves \pma positively on instances of load at most $1/3$. 
\end{theorem}
\begin{proof}
We show by induction on the size of $S$, that $\Fo(A) \leq \lvert S\rvert(3\tau -1) + \tau - 1$. For $\lvert S\rvert = 1$, it is clear since a single message forbid at most $(3\tau -1) + \tau -1 = 4\tau-2$ offsets, as explained before. Now, assume $\Fo(A) \leq \lvert S\rvert(3\tau -1) + \tau -1$ and consider a message $i \notin S$ such that \firstfit builds $A[i \rightarrow o]$ from $A$. By definition of \firstfit, choosing $o-1$ as offset creates a collision. W.l.o.g. say it is a collision in the first period. It means that there is a scheduled message between $o - \tau $ and $o-1$, hence all these offsets are forbidden by $A$. The same offsets are also forbidden by the choice of $o$ as offset for $i$, hence at most $3\tau -1$ new offsets are forbidden, that is $\Fo(A[i \rightarrow o]) \leq \Fo(A) + (3\tau -1)$, which proves the induction. 
The value of $\Fo(A)$ is increasing during \firstfit, hence it is maximal when the last element is scheduled. It is then bounded by $(n-1)(3\tau -1) + \tau - 1 < 3n\tau$.
Therefore, \firstfit succeeds when $\Fo(A) < P$, which happens when $3n\tau \leq P$, i.e. when the load is at most $1/3$.
\end{proof}

A naïve implementation of \firstfit is in time $O(nP)$, since for each of the $n$ messages $P$ offsets could be tried. However, it is not useful to consider every possible offset at each step. By maintaining a list of increasing positions of scheduled messages in the first and second period, 
we can skip all positions corresponding with a collision with the same message, and only $O(n)$ offsets must be considered. Hence, \firstfit can be implemented in time $O(n^2)$.

\subsection{Meta-Offset}

The method of this section is described in~\cite{bartharxiv2018deterministic} and it achieves the same bound on the load using a different method. It is recalled here as an introduction to the algorithms of the next section.
The idea is to restrict the possible offsets at which messages can be scheduled. It seems counter-intuitive since it decreases artificially the number of available offsets to schedule new messages. However, it allows reducing the number of forbidden offsets for unscheduled messages. A \textbf{meta-offset} is an offset of value $i\tau$,
with $i$ an integer from $0$ to $\lceil P / \tau \rceil - 1$. We call \metaoffset the greedy algorithm which works as \firstfit, but consider only meta-offsets when scheduling messages. 

To study \metaoffset, we introduce a variant of $\Fo(A)$ restricted to meta-offsets.
 Let $A$ be a partial assignment defined over $S$ and $i\notin S$, we let $\Fmo(A)$ be the maximum over $i \in [n]$ of $\lvert\left\{ j \in [\lceil P / \tau \rceil] \mid A[i \rightarrow j\tau] \text{ has a collision}\right\}\rvert$.
 By definition, two messages with a different meta-offset cannot collide in the first period. Hence, $\Fmo(A)$ can be bounded by $3\lvert S\rvert$ and we obtain the following theorem.

\begin{theorem}[Proposition 3 of~\cite{bartharxiv2018deterministic}]\label{th:metaoffset}
\metaoffset solves \pma positively on instances of load at most $1/3$.
\end{theorem}

The complexity of \metaoffset is in $O(n^2)$, since for each of the $n$ messages at most $3n$ meta-offsets must be checked in constant time. 

\subsection{Compact Pairs}

We present in this section \texttt{Compact k-tuples}, a new family of greedy algorithms which solve \pma positively for larger loads. The idea is to schedule several messages at once, using meta-offsets, to maximize the compactness of the obtained solution. We first describe an algorithm that schedules pairs of messages, which is then extended to any tuple of messages. 

For clarity of exposition, we assume from now on that \emph{the period $P$ is a multiple of $\tau$}: we let $P = m\tau$. We show in Sec.~\ref{sec:gen} that this hypothesis can be relaxed to the price of a very small increase of the load. 

When all delays are multiple of $\tau$, \metaoffset schedules the messages compactly and solves $\pma$ for a load of $1/2$. Hence, we are interested in the remainder modulo $\tau$ of the delays. Let $d_i = d_{i}'\tau + r_i$ be the Euclidean division of $d_i$ by $\tau$. We now assume that \emph{messages are sorted by increasing $r_i$}.

The \textbf{gap} between message $i$ and message $j$, is defined as $g = d'_{i} + 1 - d'_{j} \mod m$.
A \textbf{Compact pair} is a pair of messages $(i,j)$, with $i < j$ and their gap is different from $0$. A compact pair $(i,j)$ \emph{is scheduled as a single message} using meta-offsets so that $A(i) + (d'_i+1)\tau = A(j) + d'_j\tau$, i.e. the beginning of $j$ is less than $\tau$ unit of times after the end of $i$ in the second period, see Fig.~\ref{fig:compactpair}. The gap is interpreted as the distance in meta-offsets between $i$ and $j$ in the first period, when they are scheduled as a compact pair.

\begin{figure}[h]
\begin{center}

\includegraphics[scale=0.7]{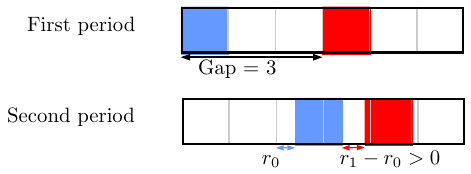}
\end{center}
\caption{The compact pair $(0,1)$ with $d'_0 = 2$ and $d'_1 = 0$}
\label{fig:compactpair}
\end{figure}

\begin{lemma}\label{lemma:pair_find}
Given three messages $(1,2,3)$ in order of increasing delay modulo $\tau$, then either $(1,2)$, $(1,3)$ or $(2,3)$ is a compact pair. 
\end{lemma}
\begin{proof}
If the first two messages or the first and the third message form a compact pair, we are done. If not, then by definition $d_{1}' = 1 + d_{2}' = 1 + d_{3}'$. Hence, messages $2$ and $3$ have the same delay divided by $\tau$ and form a compact pair of gap $1$.
\end{proof}

Let \compactpair be the following greedy algorithm:  A sequence of at least $n/3$ compact pairs is built by considering triples of messages in order of increasing $r_i$, and applying Lemma~\ref{lemma:pair_find} to each triple. Compact pairs are scheduled in the order they have been built at the first available meta-offset. If at some point all compact pairs are scheduled or the current one cannot be scheduled, the remaining messages are scheduled as in \metaoffset. 

The analysis of \compactpair relies on the evaluation of the number of forbidden meta-offsets. In the first phase of \compactpair, we evaluate the number of forbidden offsets when scheduling any compact pair, that we denote by $\Fmo_2(A)$. In the second phase, we need to evaluate $\Fmo(A)$. When scheduling a message in the second phase, a scheduled compact pair only forbids \emph{three} meta-offsets in the second period, while two messages scheduled independently forbid \emph{four} meta-offsets, which explains the improvement from \compactpair. We state the previous fact as Lemma~\ref{lemma:pair_forbid}, see an illustration in Fig.~\ref{fig:forbidenmeta}. 

\begin{lemma}\label{lemma:pair_forbid}
Let $C_1$ be compact pair. Let $C_2$ be a compact pair and let $i$ be a single message, both scheduled by \compactpair after $C_1$.
Then, because of collisions in the second period, $C_1$ forbids at most four meta-offsets to $C_2$ and three meta-offsets to $i$. 
\end{lemma}

\begin{figure}
\begin{center}
\includegraphics[scale=0.7]{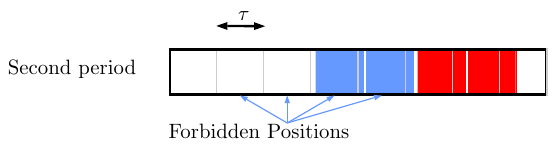}
\end{center}

\caption{Positions forbidden by a scheduled compact pair (in blue) when scheduling another compact pair (in red) with larger $r_i$'s} 
\label{fig:forbidenmeta}
\end{figure}

\begin{theorem}\label{th:pair}
\compactpair solves \pma positively on instances of load at most $3/8$.
\end{theorem}
\begin{proof}
Let $n$ be the number of messages to schedule.
In the first phase of \compactpair, $n_2$ compact pairs are scheduled to build a partial assignment $A$. Let us first assume that there are more compact pairs built using the method of Lemma~\ref{lemma:pair_find} than compact pairs which are scheduled.
When scheduling a new compact pair, the position of the first $n_2$ compact pairs on the first period forbid $4n_2$ offsets for a compact pair. In the second period, we use Lemma~\ref{lemma:pair_forbid} to bound the number of forbidden offsets by $4n_2$.   
Hence, $\Fmo_2(A) \leq 8n_2$. By definition of $\Fmo_2(A)$, there are offsets to schedule compact pairs while $\Fmo_2(A) < m$. Thus, $A$ can be extended by compact pairs if $8n_2 < m$.
Since exactly $n_2$ pairs are scheduled, $n_2 \geq m/8$.  

Let $n_1$ be the number of messages remaining to be scheduled by \compactpair using meta-offsets after the $n_2$ compact pairs have been scheduled. By Lemma~\ref{lemma:pair_forbid}, a compact pair forbids $3$ meta-offsets in the second period. It also forbids $2$ meta-offsets in the first period. The $n_1$ messages scheduled as in \metaoffset forbids $3$ meta-offsets each, as explained in the previous section.
Therefore, we have $\Fmo(A) \leq 5n_2 + 3n_1$. 

Assume now that \compactpair fails to schedule the $n_1$ messages of the second phase.
Since \compactpair can schedule new messages when $\Fmo(A) < m$, 
we obtain $5n_2 + 3(n_1-1) \geq m$. We have already proved $n_2 \geq  m/8$, by summing both inequations, we obtain $$6n_2 + 3n_1 > 9m/8 + 2.$$ By definition $n = 2n_2 + n_1$, hence $$3n > 9m/8.$$ Therefore, when \compactpair fails, $n > (3/8)m$, i.e. the load is larger than $3/8$.

Let us now assume that the first phase stops because the algorithm runs out of compact pairs. They are built using Lemma~\ref{lemma:pair_find}, for each triple of messages at least a compact pair is produced. Hence, we obtain at least $n/3$ compact pairs, which are all scheduled. We have $n_2 > n/3$ and $n= 2n_2 + n_1$, which implies that $n_1 + n_2 \leq 2n/3$.
Assume now that \compactpair fails to schedule all messages, then $5n_2 + 3n_1 > m$.
By substitution in the previous inequation, we have $ 8n/3  > m$, i.e the load is larger than $8/3$, which proves the theorem.
\end{proof}

\subsection{Compact Tuples}\label{sec:compact}

Algorithm \compactpair can be improved by forming \emph{compact tuples} instead of compact pairs. Recall that a delay $d_i$ is equal to  $d'_i\tau + r_i$, we call $d'_i$ the \textbf{meta-delay}. The algorithm we describe relies only on meta-delays and the fact that values $r_i$ are increasing.

\begin{definition}
Let $i_1 < \dots < i_k$ be a sequence of messages with $r_{i_1},\dots,r_{i_k}$ increasing. 
It is a compact $k$-tuple, if there is a valid partial assignment of these messages, such that messages in the second period are in order $i_1,\dots,i_k$ and for all $l$, $A(i_l) + (d'_{i_l} + 1)\tau \mod P = A(i_{l+1}) + d'_{i_{l+1}}\tau \mod P$.
\end{definition}

 Scheduling a compact $k$-tuple is choosing a meta-offset for the first message of the tuple, the offsets of the other messages are also fixed by this choice. The algorithm \texttt{Compact k-tuples} works by scheduling compact $k$-tuples using meta-offsets while possible, then scheduling compact $(k-1)$-tuples and so on until $k=1$.

\begin{lemma}\label{lemma:uple_find}
Given $k + k(k-1)(2k-1)/6$ messages, $k$ of them always form a compact $k$-tuple and we can find them in time $O(k^3)$. 
\end{lemma}
\begin{proof}
We assume the messages are sorted by increasing $r_i$'s. We prove the lemma by induction on $k$. Lemma~\ref{lemma:pair_find} already proves the lemma for $k=2$.
Now assume that we have found $C$ a compact $(k-1)$-tuple in the first $k-1 + (k-1)(k-2)(2k-3)/6$ messages. Consider the next $(k-1)^2 + 1$ messages: if $k$ of them have the same meta-delay, then they form a compact $k$-tuple and we are done. Otherwise, there are at least $k$ different meta-delays in those $(k-1)^2 + 1$ messages. When we want to add an additional $k$-th message in $C$, to obtain a compact $k$-tuples, the $k-1$ elements of $C$ forbid each one possible meta-delay. By pigeonhole principle, one of the $k$ messages with distinct meta-delays can be used to extend $C$. We can thus build a compact $k$-tuple from at most $(k-1) + (k-1)(k-2)(2k-3)/6 + (k-1)^2 + 1$ messages, that is $k + k(k-1)(2k-1)/6$ messages which proves the induction.
\end{proof}

In \compactpair, the compact pairs are created and scheduled in order of their $r_i$. In \texttt{Compact k-tuples}, the compact tuples are produced thanks to~\ref{lemma:uple_find} by going over the set of messages several times, hence inside a compact tuple the $r_{i_1}$'s are increasing but not
between successive tuples.

\begin{theorem}\label{th:k-tuples}
\texttt{Compact 8-tuples} always solves \pma positively on instances of load at most $2/5$ and with more than $205$ messages.
\end{theorem}
\begin{proof}
We use the following fact, which generalizes Lemma~\ref{lemma:pair_forbid}: A $k$-tuple forbids $k+j+1$ offsets in the second period when scheduling a $j$-tuple. Remark that one more offset is forbidden compared to Lemma~\ref{lemma:pair_forbid} when $k = j = 2$, because it is not true anymore that all
$r_i$ in the $j$-tuple are larger than those of the already scheduled $k$-tuple.

Let us denote by $n_i$ the number of compact $i$-tuples scheduled by the algorithm. We now compute a lower bound on the $n_i$ for $i$ equal $k$ down to $1$ by bounding $\Fmo_i(A)$, the number of forbidden meta-offsets when scheduling compact $i$-tuples in the algorithm. 
We have the following equation:  $$ \Fmo_i(A) \leq \displaystyle{\sum_{j=i}^k n_j(j+1)*(i+1)}.$$
The equation for $n_1$ is slightly better: 
$$ \Fmo(A) \leq \displaystyle{\sum_{j=1}^k n_j(2j + 1)}.$$
A lower bound on $n_i$ can be computed, using the fact that $A$ can be extended while $\Fmo_i(A) < m$ and assuming we know the value of the $n_j$'s with $j > i$. 
Lemma~\ref{lemma:uple_find} ensures that enough compact $i$-tuples can be built, when $n + n_i - \sum_{i \leq j \leq 8} j*n_j$ is larger than $i + i(i-1)(2i-1)/6$. 
A numerical computation of the $n_i$'s shows that \texttt{Compact 8-tuples} always finds an assignment when the load is at most $2/5$ and for $n \geq 205$.
\end{proof}

The code computing the $n_i$'s and thus the bound on the load can be found on github\footnote{\url{https://github.com/Mael-Guiraud/GuiraudStrozecki2023Scheduling}}. Th.~\ref{th:k-tuples} is obtained for $k=8$. Taking arbitrary large $k$ and using refined bounds on $\Fmo_i(A)$ is not enough to get an algorithm working for a load of $41/100$ (and it only works from larger $n$). To produce a compact $8$-tuples by Lemma~\ref{lemma:uple_find}, there must be $148$ messages, hence the restriction of $n \geq 205$ to be able to produce enough compact $8$-tuples.
The bound of Lemma~\ref{lemma:uple_find} can be improved, by using more complex algorithms to construct $k$-tuples, e.g. a simple case analysis shows that, in the worst case, $7$ messages are necessary to construct a compact $3$-tuple and not $8$. However, the restriction on $n$ is not relevant in practice, since on random instances, the probability that $k$ messages do not form a compact $k$-tuples is low, and thus we can build the $k$-tuples greedily. For instance, with $P=50\tau$ and thus $n \leq 50$, there is a probability larger than $55\%$ that $8$ random messages form a compact $8$-tuples, $86\%$ for $9$ messages and $96\%$ for $10$ messages.

\section{Messages of Size One} \label{sec:small}

We consider in this section the special case $\tau = 1$. While $\tau > 1$ for a C-RAN application, other applications such as sensors communicating with a base station through a low bandwidth channel may be modeled with $\tau = 1$. Moreover, in Sec.~\ref{sec:reduction}, we prove that any instance with $\tau >1$ can be transformed into an instance with $\tau = 1$, by increasing the load or the latency of the system.

When $\tau = 1$ and the load is less than $1/2$, \emph{any greedy algorithm} solves \pma positively since $\Fo(A) \leq (4\tau -2)\lvert S\rvert = 2\lvert S\rvert$ where $S$ is the set of scheduled messages. In this section, we give a polynomial time algorithm that always finds a valid assignment when the load is less than $1/2 + (\sqrt{5}/2 -1)$. We also show that a simple randomized greedy algorithm
works for loads arbitrarily close to one on random instances.

\subsection{Deterministic Algorithm}

We define the notion of \emph{potential} of a partial assignment, which indirectly measures how many offsets are left available by this assignment for all messages of the instance. To go above $1/2$ of load, we introduce the the \swapandmove algorithm: it schedules message greedily, and when it fails, the potential is optimized by local operations on the partial assignment, to get more available offsets. 

\begin{definition}
Let $i$ be a message of delay $d$ and let $A$ be a partial assignment.
The potential of $i$ for $A$, denoted by $pot_{msg}(i)$, is the number of integers $p \in [P]$ such that $p$ is used in the first period and $p+d \mod P$ is used in the second period. 
\end{definition}

The computation of the potential of a message of delay $3$, is illustrated in Fig.~\ref{fig:messagepotential}. The potential of a message counts how many forbidden offsets are avoided by the message given a partial assignment $A$.
Indeed, when $p$ is used in the first period and $p+d \mod P$ is used in the second period, then the same offset is forbidden \emph{twice} for a message of delay $d$. Hence, the potential of a message is related to the number of possible offsets as stated in the following lemma. 
\begin{figure}
\begin{center}
\includegraphics[scale=1]{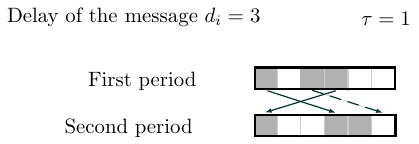}
\end{center}
\caption{A partial assignment $A$ with scheduled messages in gray. Message $i$ of delay $3$ has potential $2$ because of positions $0$ and $3$ in the first period.}
\label{fig:messagepotential}
\end{figure}

\begin{lemma}\label{lemma:pot_msg}
Given a partial assignment $A$ of size $s$, and a message $i$, then the set $\{o \mid A[i \rightarrow o] \text{ has no collision}\}$ is of size $P - 2s + pot_{msg}(i)$.
\end{lemma}
\begin{proof}
 Each of the $s$ messages scheduled by $A$ forbids at most two offsets for $i$, that is $2s$ in total. In these $2s$ forbidden offsets, exactly $pot_{msg}(i)$ are counted twice by definition of the potential of the message $i$. Since there are $P$ possible values for the offset $o$ of the message $i$, there are $P - 2s + pot_{msg}(i)$ of these values which do not create a collision when scheduling $i$.\end{proof}

We define a global measure of the quality of a partial assignment. 
Given a partial assignment $A$, the sum of potentials of all messages in the instance is called \textbf{the potential} of the assignment $A$ and is denoted by $Pot(A)$.

\begin{definition}
Let $p \in [P]$ be a position in the first period, and let $A$ be a valid partial assignment. The potential of $p$, denoted by $pot_{pos}(p)$, is the number of messages $i \in [n]$, such that there is a collision in $A[i \rightarrow p]$. 
\end{definition}

\begin{figure}
\begin{center}
\includegraphics[scale=1]{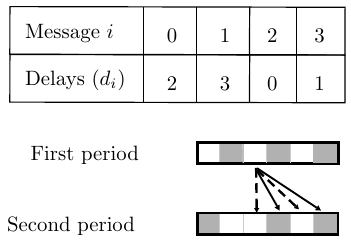}
\end{center}
\caption{A partial assignment $A$ with scheduled messages in gray. Position $p=2$ is of potential $2$ because there is a collision in $A[1 \rightarrow p]$ and  $A[3 \rightarrow p]$.}
\label{fig:positionpotential}
\end{figure}

The potential of a position is illustrated in Fig.~\ref{fig:positionpotential}. Instead of decomposing the global potential as a sum over messages, it can be interpreted as a sum over positions, as stated in the next lemma.

\begin{lemma}\label{lemma:pot_pos}
The sum of potentials of all positions used in the first period by messages scheduled by $A$ is equal to $Pot(A)$.  
\end{lemma}
\begin{proof}
 Let us denote by $\chi_{o,d}(A)$ the indicator function which is equal to one if and only if for $A$, $o$ is used in the first period and $o+d$ is used in the second period.

By definition of potential of an assignment, 
$$Pot(A) = \sum_{i \in [n]} pot_{mes}(i).$$
By definition of the potential of a message, 
$$\displaystyle{Pot(A) = \sum_{i \in [n]} \sum_{\substack{o \in [P]}} \chi_{o,d_i}(A) }.$$
By exchanging the summation order, we obtain
$$\displaystyle{Pot(A) = \sum_{o \in [P] }\sum_{i \in [n]} \chi_{o,d_i}(A)}.$$

Let $O$ be the set of offsets used by $A$.
When $o$ is not in $O$, it contributes nothing to the sum, thus
$$\displaystyle{Pot(A) = \sum_{o \in O} \sum_{i \in [n]} \chi_{o,d_i}(A) }.$$
Then, by definition of potential of a position we obtain

$$\displaystyle{Pot(A) = \sum_{o \in O} pot_{pos}(o)}.$$
\end{proof}

The sum of the potentials of \emph{all} positions can easily be computed and only depends on the size of the partial assignment.

\begin{lemma}\label{lemma:inv}
The sum of potentials of all positions of a partial assignment of size $k$ is $nk$.  
\end{lemma}
\begin{proof}
We want to compute $\sum_{o \in [P]} pot_{pos}(o)$ which is by definition
$$\sum_{o \in [P]} \sum_{\substack {i \in [n],\, o + d_i \text{ is used} \\ \text{in the second period}}} 1.$$
Remark that for one position $p$ used in the second period, and one message $i$, there is exactly one offset $o$ such that $p = o+d_i \mod P$.
Hence, each used position of the second period contributes $1$ for each message in the double sum, that is $n$ in total.
Since $A$ is of size $k$, there are $k$ used positions and the sum is equal to $kn$.
\end{proof}

 As a consequence of this lemma, $Pot(A) \leq nk$. Let us define a \textbf{Swap operation}, which guarantees to obtain at least half the maximal value of the potential. Let $A$ be some partial assignment of size $s$ and let $i$ be an unscheduled message. Assume that $i$ cannot be used to extend $A$. The Swap operation is the following: 
select a free position $p$ in the first period, remove the message which uses the position $p+d_i$ in the second period from $A$, and extend $A$ by $i$ with offset $p$. We denote by $Swap(i,p,A)$ the partial assignment obtained by this operation.

\begin{lemma}\label{lemma:swap}
Let $A$ be some partial assignment of size $k$ and let $i$ be an unscheduled message. If $i$ cannot be used to extend $A$, then either $Pot(A) \geq kn/2$ or there is $p \in [P]$ such that $Pot(Swap(i,p,A)) > Pot(A)$.
\end{lemma}

\begin{proof}
The positions in the first period can be partitioned into $P_{u}$ the positions used by some scheduled message and $P_{f}$ the free positions.
Let $V_f$ be the sum of the potentials of the positions in $P_f$ and let $V_u$ be the sum of the potentials of the positions in $P_u$. By Lemma~\ref{lemma:inv}, since $P_f$ and $P_u$ partition the positions, we have $V_f + V_u = kn$. Moreover, by Lemma~\ref{lemma:pot_pos}, $Pot(A) = V_u$, then $V_f + Pot(A) = kn$.

By hypothesis, $i$ cannot be scheduled, then, for all $p \in P_{f}$, $p+d_i$ is used in the second period. Let $F$ be the function which associates to $p \in P_{f}$ the position $A(j)$ such that there is $j$ a scheduled message which uses $p+d_i$ in the second period, that is $A(j) + d_j = p + d_i \mod P$. The function $F$ is an injection from $P_{f}$ to $P_u$. Remark that, in both $Swap(i,p,A)$ and $A$, the same positions are used in the second period. Hence, the potential of each position remains the same after the swap. As a consequence, doing the operation $Swap(i,p,A)$ adds to $Pot(A)$ the potential of the position $p$ and removes the potential of the position $F(p)$. 

Assume now, to prove our lemma, that for all $p$, $Pot(Swap(i,p,A)) \leq Pot(A)$. It implies that for all $p$, the potential of $p$ is smaller than the potential of $F(p)$. Since $F$ is an injection from $P_f$ to $P_u$, we have that $V_f \leq V_u = Pot(A)$. Since $V_f + Pot(A) = kn$, we have that $Pot(A) \geq kn/2$.
\end{proof}

Let us now define algorithm \swapandmove. 
It schedules messages using \firstfit while possible. Then, it applies the Swap operation while it increases the potential. When the potential cannot be improved by a Swap anymore, \swapandmove try to schedule a new message at each position. When scheduling the message at some position, if it conflicts with one or two already scheduled messages, they are moved to another offset if possible. If \swapandmove fails to schedule the message it stops, otherwise the whole procedure is repeated. 

Algorithm \swapandmove is not greedy, since we allow to change the offset of a message, either to improve the potential or to free an offset for scheduling a new message. However, the number of scheduled messages during the algorithm increases and a message cannot be unscheduled, it only has its offset changed.
While computing the potential requires knowing all delays in advance, \swapandmove can be adapted to work online by considering the potential of a partial assignment to be the sum of potentials of the scheduled messages.

\begin{theorem}
\swapandmove solves \pma positively, in time $O(n^3)$, for instances with $\tau =1$ and load at most $(\sqrt{5}-1)/2 \approx 0,618$.
\end{theorem}

\begin{proof}
We determine for which value of the load \swapandmove always finds an assignment.
We consider the situation when $n-1$ messages are scheduled by $A$ out of $n$ and \swapandmove tries to schedule the last one. The proof that the algorithm schedules the previous messages is the same. We let $n - 1 = (1/2 + \varepsilon)P$ be the number of messages, hence the load we achieve is $1/2 + \varepsilon$. 

Let $d$ be the delay of the last unscheduled message, w.l.o.g. we assume that $d = 0$. 
Let $P_f$ be the set of $p\in [P]$ which are free in the first period. Let $P^1_{u}$ be the set of $p\in [P]$, such that $p$ is used in the first period but $p$ is free in the second period. Let $P^2_{u}$ be the set of $p\in [P]$, such that $p$ is used in the first period and $p$ is used in the second period. 

Since there are $n-1$ messages scheduled by $A$, considering the positions used in the first period, we have $n-1 = \lvert P^1_u\rvert + \lvert P^2_u\rvert$.
To prove a lower bound on the load, we assume that the last message cannot be scheduled greedily by \swapandmove. Hence, if $p\in P_f$, then $p$ is used in the second period. By considering the used positions in the second period, we have $n - 1 = \lvert P^2_u\rvert + \lvert P_f\rvert$.
Hence, because the three sets partition $[P]$, $\lvert P^2_u\rvert = 2(n-1) - P = 2\varepsilon P$ and $\lvert P^1_u\rvert = \lvert P_f\rvert = (1/2 - \varepsilon)P$.

Consider a position $p \in P^1_u$, then $p$ is used in the first period by some message $i$.
If the offset of $i$ can be changed in $A$ to obtain a valid assignment, then \swapandmove does it and succeeds because the last message can now be scheduled at offset $p$.
By Lemma~\ref{lemma:pot_msg}, the number of offsets which can be used by some message $i$ is $P- 2(n-1) + pot_{msg}(i)$. Since the offset of $i$ cannot be changed, this number is zero and we have $pot_{msg}(i) = 2(n-1) - P = 2\varepsilon P$. For the same reason, a message $i$ with $(A(i) + d_i \mod P) \in P_f$ has the same potential.

Consider a position $p \in P^2_u$, then $p$ is used by a message $i$ in the first period and $p$ is used by a message $j$ in the second period. If both $i$ and $j$ can be scheduled elsewhere, then \swapandmove moves them and succeeds. 
By Lemma~\ref{lemma:pot_msg}, both messages can be scheduled when one is of potential at least $2\varepsilon P + 1$ and the other at least $2\varepsilon P + 2$. 
Since a message is of potential at most $n-1$, both messages can always be rescheduled when the sum of the two potentials is at least $2\varepsilon P + n$. Hence, we may assume that their sum is less than $2 \varepsilon P + n$. 

By definition, $Pot(A)$ is the sum of the potential of the messages, that we can write as 
$$Pot(A) = \sum_{\substack{i \in [n] \\ A(i) \in P^1_u} } pot_{msg}(i) + \sum_{\substack{i \in [n] \\ A(i) \in P^2_u} } pot_{msg}(i).$$

We may also divide the sum according to the positions in the second period,

$$Pot(A) = \sum_{\substack{i \in [n] \\ A(i) + d_i \in P_f} } pot_{msg}(i) + \sum_{\substack{i \in [n] \\ A(i) + d_i \in P^2_u} } pot_{msg}(i).$$

By summing both equalities, using $\lvert P^1_u\rvert = \lvert P_f\rvert = (1/2 - \varepsilon)P$ and 
$pot_{msg}(i) = 2\varepsilon P$ for $A(i) \in P^1_u$ and  $A(i) + d(i) \in P_f$, we obtain  

\begin{equation*}
  \begin{aligned}
    2Pot(A) & = 2(1/2 - \varepsilon)P \times 2\varepsilon P + \sum_{\substack{i \in [n] \\ A(i) \in P^2_u} } pot_{msg}(i) \\
		  & + \sum_{\substack{i \in [n] \\ A(i) + d_i \in P^2_u} } pot_{msg}(i).
  \end{aligned}
\end{equation*}

We have proved that $$pot_{msg}(i) + pot_{msg}(j) < 2\varepsilon P + n$$ when $A(i) \in P^2_u$ and $A(i) = (A(j) + d_j \mod P)$. Since  $\lvert P^2_u\rvert = 2\varepsilon P$, we have $$\sum_{\substack{i \in [n] \\ A(i) \in P^2_u} } pot_{msg}(i) + \sum_{\substack{i \in [n] \\ A(i) + d_i \in P^2_u} } pot_{msg}(i) < 2\varepsilon P(2\varepsilon P + n).$$
Using this inequality and simplifying the expression, we obtain

$$  Pot(A) < \varepsilon P (P + n).$$

This bound is obtained when \swapandmove fails to schedule the last message.
On the other hand, by Lemma~\ref{lemma:swap}, we know that $Pot(A) \geq n(n-1)/2$, hence 
\swapandmove must succeed when
$$n(n-1)/2 \geq  \varepsilon P (P + n -1).$$
By expanding and simplifying the previous inequation, we obtain a second-degree inequation in $\varepsilon$, $1/4 - 2\varepsilon - \varepsilon ^2 \geq  0$.
Solving this inequation yields $\varepsilon \leq \sqrt{5}/2 -1$.

Let us prove that \swapandmove is in polynomial time. All Swap operations 
strictly increase the potential. Moreover, when one or two messages are moved, the potential may decrease but a message is added to the partial assignment. The potential is bounded by $O(n^2)$ and the move operations all together can only remove $O(n^2)$ to the potential, hence there are at most $O(n^2)$ Swap operations during \swapandmove. A Swap operation can be performed in time $O(n)$, since, for a given message, all free offsets must be tested and the potential is evaluated in time $O(1)$ (by maintaining the potential of each position). This proves that Swap and Move is in $O(n^3)$.  
\end{proof}

Consider a partial assignment of size $n-1 = (1/2 + \varepsilon)P$, and a last message of delay $d$ to schedule.
We have seen that if a scheduled message cannot be rescheduled, its potential is equal to $2\varepsilon P$, it is larger otherwise.
Hence, the best possible upper bound on the potential of the assignment is $2\varepsilon P n$. On the other hand, Lemma~\ref{lemma:swap} guarantees that the potential of an assignment is at least $n(n-1)/2$. Therefore, improving the analysis at best yields $\varepsilon = 1/6$ and load $2/3$.

To go further, the analysis in Lemma~\ref{lemma:swap} may be improved: $2\varepsilon P$ positions in $P_{u}$ are not taken into account in the proof. When the delays are distinct we can indeed use this remark to improve the result. However, there is an instance (found by a bruteforce search) with $8$ messages and $P=10$ for which there is no assignment, hence the largest $\lambda$ for which $\pma$ has always a solution is strictly less than $8/10$. 

\subsection{Randomized Algorithm for Random Instances}

We would like to better understand the behavior of greedy algorithms on instances drawn uniformly at random. To this aim, we analyze the algorithm \greedyuniform, defined as follows: for each message in the order of the input, choose one of the offsets, which does not create a collision with the current partial assignment, uniformly at random. 

We analyze \greedyuniform over random instances:  all messages have their delays drawn independently and uniformly in $[P]$. We compute the probability of success of \greedyuniform over all random choices by the algorithm \emph{and all possible instances}. 
It turns out that this probability, for a fixed load strictly less than one, goes to one when $P$ grows. To simplify the analysis of \greedyuniform, we introduce the notion of \textbf{trace} of an assignment. 

\begin{definition}
Let $A$ be a partial assignment of size $k$, its trace is a pair of subsets $(S_1,S_2)$ of $[P]$ of size $k$ such that $S_1$ are the time used by $A$ in the first period and $S_2$ the time used by $A$ in the second period.
\end{definition}

 We now prove that traces are produced uniformly by \greedyuniform.

\begin{theorem}\label{theorem:uniform}
The distribution of traces of assignments produced by \greedyuniform when it succeeds, from instances drawn uniformly at random, is also uniform.
\end{theorem}
\begin{proof}
The proof is by induction on $n$, the number of messages. It is clear for $n=1$,
since the delay of the first message is uniformly drawn and all offsets can be used.
Assume now the theorem is true for some $n>1$. By induction hypothesis, \greedyuniform has produced
uniform traces from the first $n$ messages. Hence, we should prove that, if we draw the delay
of the $n+1^{th}$ message randomly, extending the trace by a random possible offset produces a random distribution on the traces of size $n+1$. 

 If we draw an offset uniformly at random (among all $P$ offsets) and then extend the trace by scheduling the last message at this offset or fail, the distribution over the traces of size $n+1$ is the same as what produces \greedyuniform. Indeed, all offsets which can be used to extend the trace have the same probability to be drawn. Since all delays are drawn independently, we can assume that, given a trace, we first draw an offset uniformly, then draw uniformly the delay of the added message and add it to the trace if it is possible. This proves that all extensions of a given trace are equiprobable. Thus, all traces of size $n+1$ are equiprobable, since they each can be formed from $(n+1)^2$ traces of size $n$ by removing one used time from the first and second period. This proves the induction and the theorem.
\end{proof}

Since \greedyuniform can be seen as a simple random process on traces by Th.~\ref{theorem:uniform}, it is easy to analyze its probability of success.

\begin{theorem}
The probability over all instances with $n$ messages and period $P$ that \greedyuniform solves $\pma$ positively is $\displaystyle{\prod_{i \geq P/2}^{n-1}\left(1 - \frac{\binom{i}{2i-P}}{\binom{P}{i}}\right)}$.
\end{theorem}
\begin{proof}
We evaluate $\Pr(P,i)$ the probability that \greedyuniform fails after succeeding to assign the first $i$ messages, that is when it is not possible
to assign the $(i+1)^{th}$ message. This probability is independent of the delay of the $(i+1)^{th}$ message. Indeed, the operation which adds one to all times used in the second period is a bijection on the set of traces of size $i$. It is equivalent to removing one to the delay of the $(i+1)^{th}$ message. We can thus assume that the delay is zero.

Let $S_1$ be the set of times used in the first period by the $i$ first messages and $S_2$ the set of times used in the second period. We can assume that $S_1$ is fixed, since all subsets of the first period are equiprobable and because $S_2$ is independent of $S_1$ by Th.~\ref{theorem:uniform}. There is no possible offset for the $(i+1)^{th}$ message, if and only if $S_1 \cup S_2 = [P]$. It means that $S_2$ has been drawn such that it contains $[P] \setminus S_1$. By Th.\ref{theorem:uniform}, $S_2$ is uniformly distributed over all sets of size $i$. Hence, the probability that  $[P]  \setminus S_1 \subseteq S_2$  is the probability to draw a set of size $i$ which contains $P- i$ fixed elements. This proves $\Pr(P,i) = \frac{\binom{i}{2i - P}}{\binom{P}{i}}$.

From the previous expression, we can derive the probability of success of \greedyuniform by a simple product of the probabilities of success $(1 - \Pr(P,i))$ at step $i$, for $i \leq n$. The product is over $ P/2 \leq i < n$, because for $i < P/2$, the probability of success of \greedyuniform is one.  
\end{proof}

Let us fix the load $\lambda = n/P$. 
If we express $\Pr(P,n)$ as a function of $P$ and $\lambda$, we obtain
$$\Pr(P,n) = \frac{(\lambda P)!^2}{P!(2(\lambda -1)P)!}.$$

Using Stirling approximation, there are two positive constants $C_1$ and $C_2$ such that $$C_1 n^{1/2}\left(\frac{n}{e}\right)^n< n! < C_2n^{1/2}\left(\frac{n}{e}\right)^n.$$
Using the previous approximation, a computation yields a constant $C$ independent from $\lambda$ such that $$\Pr(P,n) \leq C \left(\frac{\lambda^{2\lambda}}{(2\lambda -1)^{2\lambda -1}}\right)^P.$$
As an illustration, with $\lambda = 2/3$, $\Pr(P,n) < 1,16 \times (0,84)^P$.

We let $f(\lambda) = \frac{\lambda^{2\lambda}}{(2\lambda -1)^{2\lambda -1}}$.
The derivative of $f$ is strictly positive for $1/2 < \lambda < 1$ and $f(1) = 1$, hence $f(\lambda) < 1$ when $\lambda < 1$. By a union bound, the probability that \greedyuniform fails is bounded by the sum of probabilities that it fails at step $i$ for $P/2 \leq i < n$. We have $\Pr(P,i) < \Pr(P,j)$
if $i < j$, hence we bound the probability that \greedyuniform fails by $\frac{C \lambda P f(\lambda)^P}{2}$. For any fixed $\lambda$, the previous expression goes to zero, exponentially quickly, when $P$ goes to infinity. It explains why \greedyuniform is good in practice for large $P$, even when the load is large. 
For instance, with $n = 8$ and $P=12$ (thus $\lambda = 2/3$), the probability of success of \greedyuniform is larger than $0.92$ (see Sec.~\ref{sec:perf_small} for more values).

\section{Generalizations}\label{sec:gen}

In this section, we prove that several hypotheses made in the previous sections can be relaxed and that the problem \pma captures periodic message scheduling in more complex networks.

\subsection{Period Multiple of the Message Size}

We prove that we can assume that $P$ is a multiple of $\tau$ in Lemma~\ref{lemma:multiple}. This hypothesis was done to make the analysis of algorithms based on meta-offsets simpler and tighter. 

\begin{lemma}\label{lemma:multiple}
Let $I$ be an instance of \pma with $n$ messages of size $\tau$, period $P$ and $m = \lfloor P / \tau \rfloor$. There is an integer $\tau'$ and an instance $J$ with $n$ messages of size $\tau'$ and period $P'= m\tau'$ such that any assignment of $J$ can be transformed into an assignment of $I$ in polynomial time.
\end{lemma}
\begin{proof}
Fig.~\ref{fig:multipleperiod} illustrates the reductions we define in this proof on a small instance.
Let $P = m \tau + r$ with $r < \tau$. We define the instance $I'$ as follows: $P' = mP$, $d_{i}' = m d_i$ and $\tau' = m \tau + r$. With this choice, we have $P' = m(m \tau + r) = m \tau'$.
Consider an assignment $A'$ of the instance $I'$. We let $\tau'' = m\tau$, then $A'$ is also an assignment for $I'' = (P',\tau'',(d_{0}',\dots,d_{n-1}'))$. Indeed, the size of each message, thus the intervals of time used in the first and second period begin at the same position but are shorter, which cannot create collisions. 

We consider the assignment $A'$ seen as an assignment of $I''$ and denote it by $A_0'$.
We describe a compactification procedure that produces a sequence of assignments from $A_0'$, such that, in the last assignment, any message has a position multiple of $m$ in the first and second period. See Th.4 of~\cite{bartharxiv2018deterministic} for a similar method, used to design an exponential time algorithm to solve \pma. 

W.l.o.g., we assume that message $0$ is at offset zero in $A_0'$. The first time message $0$ uses in the second period is a multiple of $m$ since its delay is by construction a multiple of $m$. Consider the following shift of $A_0'$: $A_1'(0) = 0$ and for $i>0, A_1'(i) = A_0'(i) - s$. We let $s$ be a non negative integer such that $A_0'$ shifted by $s$ is a valid assignment, while $A_0'$ shifted by $s+1$ has a collision involving message $0$. By construction of $A_1'$, because of the choice of $s$, there is a message $j$ which is next to message $0$ in the first or second period. It implies that either $A_1'(j)$ or $A_1'(j)+d_j \mod P'$ is a multiple of $m$ and since $d_j$ is a multiple of $m$, then both $A_1'(j)$ and $A_1'(j)+d_j \mod P'$ are multiples of $m$. The procedure is repeated, we obtain $A_{i+1}$ from $A_i$ by fixing $i$ messages of $A_i$ with an offset multiple of $m$ and shifting the other messages as previously. 

By construction, $A_n'$ is a valid assignment of $I''$, and all positions of messages in the first and second period are multiples of $m$. Finally, we let $A$ be the assignment of $I$ defined as, for all $i \in [n]$, $A(i) = A_n'(i)/m$. 
\end{proof}

Notice that, the transformation of Lemma~\ref{lemma:multiple} does not give a bijection between assignments of both instances but only an injection, which is enough for our purpose. 
We obtain an instance with a period $P'= m\tau'$, which is slightly smaller than the original period $P$. The load increases from $\lambda = n \tau / P$ to at most $\lambda (1 + 1/m)$: the difference is less than $1/m < 1/n$, thus very small for most instances. It corresponds to at most one meta-offset in the computation of $\Fmo$ and $\Fmo_i$. When using these functions in Sec.~\ref{sec:large}, to prove that \metaoffset, \compactpair and \texttt{Compact k-tuples} always work for some load, the computations are not tight by an additive factor of two or three. It compensates for the cost of Lemma~\ref{lemma:multiple}, thus Theorem~\ref{th:metaoffset}, Theorem~\ref{th:pair} and Theorem~\ref{th:k-tuples} still hold when $P$ is not a multiple of $\tau$.

\begin{figure*}
 \begin{center}
\includegraphics[scale=0.75]{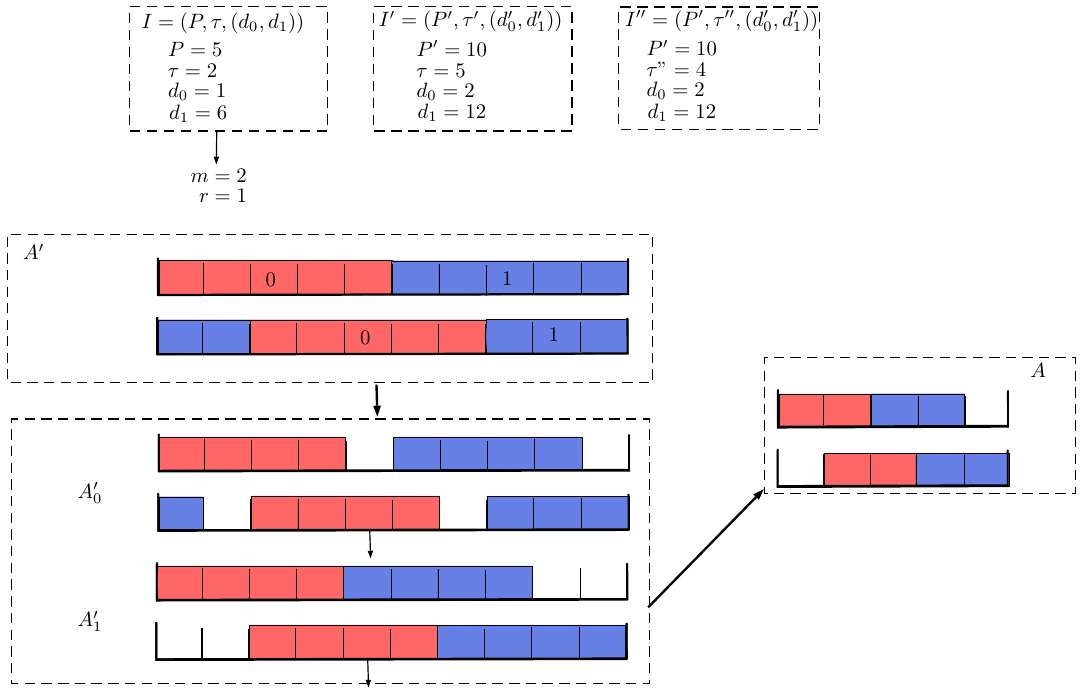}
\end{center}
\caption{Transformation of $A'$ assignment of $I'$ into $A$ assignment of $I$}
\label{fig:multipleperiod}
\end{figure*}

\subsection{From Large to Small Messages}\label{sec:reduction}

In this section, we give methods to reduce the size of messages up to $\tau = 1$, at the cost of increasing the load or allowing buffering in the network. This further justifies the interest of Sec.~\ref{sec:small}, where specific algorithms for $\tau = 1$ are given.

\paragraph*{Doubling the Load}

We describe here a reduction from an instance of \pma to another one with the same period and the same number of messages but the size of a message is doubled. This instance is equivalent to an instance with $\tau = 1$, by dividing everything by the message size. Thus, we can always assume that $\tau = 1$, if we are willing to double the load. In the following Theorem, for simplicity, we make the hypothesis that $P$ is a multiple of $2\tau$, but it can be removed using Lemma~\ref{lemma:multiple}.

\begin{theorem}\label{th:double_load}
Let $I$ be an instance of \pma with $n$ messages, load $\lambda$ and $P=m2\tau$ with $m$ an integer. There is an instance $J$ with $n$ messages of size $1$ and load $2\lambda$ such that an assignment of $J$ can be transformed into an assignment of $I$ in polynomial time.
\end{theorem}
\begin{proof}
From $I = (P,\tau,(d_{0},\dots,d_{n-1}))$, we build $I' = (P, 2\tau, (d_{0}',\dots,d_{n-1}'))$, where $d_i' = d_{i} - (d_{i} \mod 2\tau)$. The instance $I'$ has a load twice as large as $I$.
By construction, all delays of $I'$ are multiples of $2\tau$ and $P = m2\tau$.
Hence, solving \pma on $I'$ is equivalent to solving it on $J = (P/2\tau, 1,(d_{0}/ 2\tau,\dots,d_{n-1} /2\tau))$, as already explained in the proof of Lemma~\ref{lemma:multiple}. 

Let us prove that an assignment $A'$ of $I'$ can be transformed into an assignment $A$ of $I$. 
Consider the message $i$ with offset $A'(i)$, it uses all times between $A'(i)$ and $A'(i) + 2\tau -1$ in the first period and all times between $A'(i) + d_{i} - (d_{i} \mod 2\tau)$ to $A'(i) + 2\tau -1+ d_{i} - (d_{i} \mod 2\tau)$ in the second period. 
If $d_{i} \mod 2\tau < \tau $, we set $A(i) = A'(i)$, and the message $i$ of $I$ is scheduled ``inside'' the 
message $i$ of $I'$, see Fig.~\ref{fig:transf_2tau}. If $\tau \leq d_{i} \mod 2\tau < 2\tau$, then we set 
$A(i) = A'(i) - \tau$. There is no collision in the assignment $A$, since all messages in the second period use times which are used by the same message in $A'$. In the first period, the messages scheduled by $A$ use either the first half of the same message in $A'$ or the position $\tau$ before, which is either free in $A'$ or the second half of the times used by another message in $A'$ and thus not used in $A$. 
\end{proof}
\begin{figure*}[h]
\begin{center}

\includegraphics[scale=0.7]{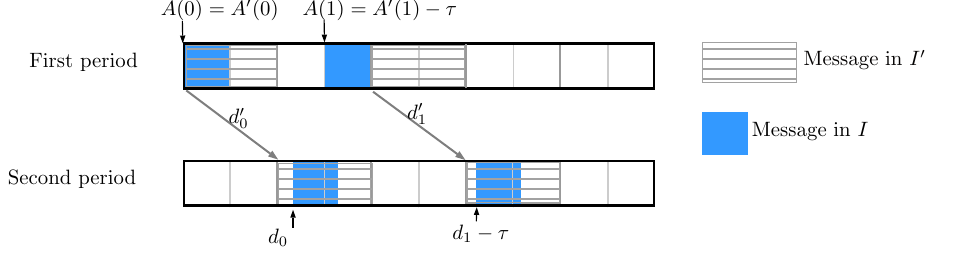}
\end{center}
\caption{Building assignment $A$ of $I$ from assignment $A'$ of $I'$ with messages twice larger}
\label{fig:transf_2tau}
\end{figure*}

By combining \greedyuniform and the transformation of Th.~\ref{th:double_load} we obtain a randomized algorithm to solve $\pma$. On random instances, it solves \pma positively with probability one when the number of messages goes to infinity and the load is strictly less than $1/2$. This is why we have not presented nor analyzed an algorithm designed for arbitrary $\tau$ on random instances, since any greedy algorithm, relying on minimizing $\Fo(A)$, cannot guarantee anything for load larger than $1/2$. However, in Sec.~\ref{sec:perf_large}, we present \compactfit, a simple greedy algorithm that exhibits good performance on random instances.

\paragraph*{Trade-off between Latency and Message Size}

The problem \pma is a simplified version of the practical problem we address, with only a single degree of freedom for each message: its offset. We may relax it slightly to be closer to what is studied in~\cite{barth2018deterministic}: we allow buffering a message $i$ during a time $b_i$ between the two contention points, which corresponds here to changing $d_i$ into $d_i + b_i$. The quality of a solution obtained for a modified instance of \pma is worse since the buffering adds \emph{latency} to the messages. This section aims to transform a given instance of \pma into a new instance with smaller messages, while minimizing the latency.

The transformation is the following: let $b_i$ be such that 
 $b_i + d_i = 0 \mod \tau$, hence the delays are multiples of $\tau$. Assuming $P = m\tau$, we have an easy reduction to the case of $\tau = 1$, by dividing all values by $\tau$, as explained in the proof of Lemma.~\ref{lemma:multiple}. 

If the $b_i$ are taken as small as possible, the largest $b_i$ may be equal to $\tau -1$, which is not so good in practice, since algorithms optimizing the latency do better on random instances, see~\cite{barth2018deterministic}. However, it is much better than buffering for a time $P$, the only value for which we are guaranteed to find an assignment, whatever the instance.

We can do the same transformation by buffering all messages so that $d_i$ is a multiple of $\tau / k$. The cost in terms of latency is then at most $\tau / k - 1$ but the reduction yields messages of size $k$. For small size of messages, it is easy to get better algorithm for \pma, in particular for $\tau = 1$ as we have shown in Sec.~\ref{sec:small}. Here, we show how to adapt \compactpair to the case of $\tau = 2$, to get an algorithm working with a higher load.

\begin{theorem}
\compactpair on instances with $\tau =2$ always solves \pma positively on instances of load at most $4/9$.
\end{theorem}
\begin{proof}
We assume w.l.o.g that there are less message with even $d_i$ than odd $d_i$.
We schedule compact pairs of messages with even $d_i$, then we schedule single messages with odd $d_i$. The worst case is when there is the same number of the two types of messages. In the first phase, if we schedule $n/2$ messages, the number of forbidden offsets is $(2 + 3/2)n/2 = 7n/4$. In the second phase, if we schedule $n/2$ additional messages, the number of forbidden offsets is bounded by $ (1 + 3/2) n/2  + (1 + 1)n/2 = 9n/4$. Hence, both conditions are satisfied and we can always schedule messages when $n \leq (4/9)m$.
\end{proof}


Alternatively, we may minimize the average latency rather than the worst latency. We show how to do the previous transformation yielding $\tau=1$ while optimizing the average latency.
In the transformation, we have chosen $b_i$ so that $b_i + d_i = 0 \mod \tau$.
However, we can also fix $t$ such that $b_i + d_i = t \mod \tau$.
By subtracting $t$ to all delays, we obtain an equivalent instance where 
all delays are multiple of $\tau$ and we can conlude as before.

 For each $i$ and $t$, we let $b_{i,t}$ be the minimal integer satisfying $d_i + b_{i,t} = t \mod \tau$. We let $L(t)$ be the sum of buffering times used for the messages when $t$ is chosen as remainder, that is $L(t) = \sum_{i=0}^{n-1} b_{i,t}$. To minimize the average latency, we must minimize $L(t)$.

 For $t \in \tau$, $b_{i,t}$ takes all possible values in $[\tau]$,
 thus, for all $i$, $\sum_{t=0}^{\tau-1} b_{i,t} = \sum_{j=0}^{\tau-1} j$, that is $\tau (\tau-1)/2$. Since there are $n$ messages, $\sum_{t=0}^{\tau -1} L(t) = n \tau (\tau-1)/2$. There is at least one term of the sum less than its average, hence there is a $t_0$ such that $L(t_0) \leq n (\tau-1)/2$. Hence, if we choose $t= t_0$ in the transformation, the average buffering of a message is less than $(\tau -1)/2$.

\subsection{Coherent Routing}\label{sec:coherent}

In this section, we explain how algorithms solving \pma may be used on more complex networks with more than two contention points.
We consider networks with \textbf{coherent routing}, a common property of telecommunication networks (see e.g.~\cite{Schwiebert1996ANA}).
Each message follows a directed path from an antenna to the data center.
The coherent routing property implies that two routes share either nothing or a single path (i.e. a sequence of contiguous links) in the network. See Fig.~\ref{fig:coherent} for a network with and without coherent routing.

\begin{figure}
\begin{center}

\scalebox{0.5}{
\includegraphics[]{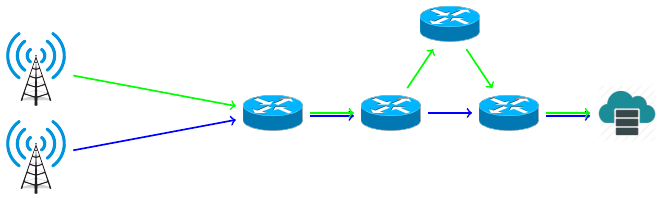}
}

\scalebox{0.5}{
\includegraphics[]{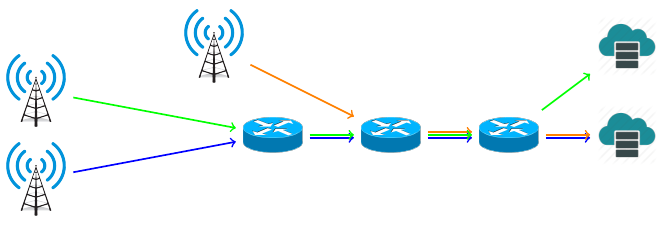}
}

 \caption{Top, a C-RAN network with a routing not coherent, and bottom a C-RAN network with coherent routing. The edges with multiple colors represent a single physical link used by several messages.}

\label{fig:coherent}
\end{center}
\end{figure}

The problem $\pma$ can be generalized to any network: the problem is to find an assignment (an offset for each message) such that there is no collision in the network. The problem $\pma$ for general networks (under the name \textsc{PAZL}) is proven to be $\NP$-hard in~\cite{bartharxiv2018deterministic} and the instances used for the reduction have a coherent routing. 

When the routing is coherent, for each pair of messages, there is either no common contention point or a single contention point which is the beginning of their common path. Indeed, if there is no collision at this contention point, there is no collision between these two messages further in the network.
In our model, the routing is coherent from the antennas to the data centers, and then it is coherent from the data centers to the antenna. Hence, the collision between two messages can be characterized by their two contention points (on the way forward and on the way back) and the delay between these two points. 

We now prove that we can transform a network with coherent routing into a network with a single contention point such that an assignment of the latter is also an assignment of the former. Therefore, the algorithms proposed in this article for the single shared link case can be transferred to the coherent routing case. 

We now describe a transformation removing a contention point. Let us consider two contention points $c_1$ and $c_2$ in the network $N$, such that there is no contention point before $c_1$ nor between $c_1$ and $c_2$ and there is an arc from $c_1$ to $c_2$ in the path of one message. Since $N$ is with coherent routing, there is a single arc between $c_1$ and $c_2$, and we denote its length by $l$. 
Let $S$ be the set of messages going through going through $c_2$ but not $c_1$. We crete a new network $N'$ by modifying the paths followed by each message $m \in S$. If $m$ follows the arc $(c,c_2)$ of length $l_m$, then we replace it by an arc $(c,c_1)$ of length $l_m - l \mod P$ followed by the arc $(c_1,c_2)$ of length $l$.

Consider a valid assignment $A$ for $N'$. By construction, the constraints to satisfy because of collisions are the same as in $N$, except on vertex $c_1$, where they are strictly stronger. Hence, the assignment $A$ is also valid for $N$. Moreover, if the network is with coherent routing, then the modified network is also with coherent routing.    

Since $N'$ is with coherent routing, there is a single arc between $c_1$ and $c_2$ of length $l$. By construction, no path arrives in $c_2$. For each message $m \in S$, with in its path the arcs $(c_1,c_2)$ and $(c_2,c_3)$ of length $l_m$, we replace the two arcs by $(c_1,c_3)$ of length $ l + l_m \mod \tau$. 
 The contention point $c_2$ is then removed from the network, as in Fig.~\ref{fig:transformation} and we obtain a network with coherent routing and the same satisfying assignments as the original one.

If the network is not with coherent routing, then the transformation would fail at the previous step. Indeed, we could have $c_1$ and $c_2$ with several
arcs of different lengths inbetween, whith no way to replace paths of size two
through $c_1$ by an arc and obtain an equivalent network. 

\begin{figure}
\begin{center}

\scalebox{0.6}{
\includegraphics[]{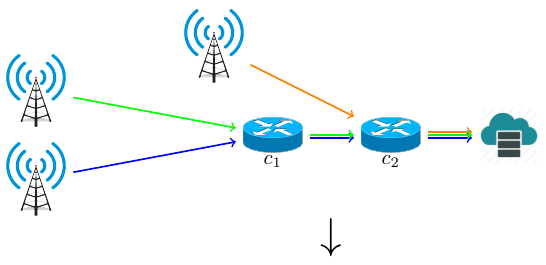}
}

\scalebox{0.6}{
\hspace{2cm}
\includegraphics[]{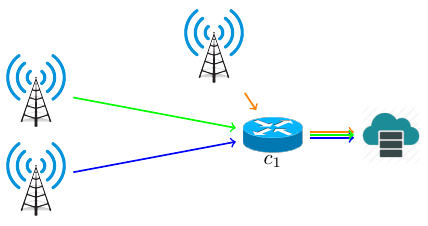}
}

 \caption{Transformation of a network $N$ with coherent routing into $N'$ with one less contention point.}

\label{fig:transformation}
\end{center}
\end{figure}

 We can apply the algorithms presented in this article to find an assignment for $N'$, which can then be turned back into an assignment of $N$. Let us define the load of a general network as $n\tau /P$, then the load of $N'$ is the same as the load of $N$. Hence, we have proved that \pma for general networks with coherent routing can be solved positively when the load is less than $2/5$ or less than $(\sqrt{5}-1)/2$ and $\tau = 1$. 

However, there is a more relevant way to define the load for general networks. Let $n_c$ be the number of messages going through the contention point $c$, then the load at $c$ is $n_c\tau/P$. The \textbf{local load} of the network is then defined as the maximum of the load of the contention points. It is always less than the load defined as $n\tau /P$ and it can be significantly so. The algorithms of this article do not seem to work for a bounded local load, since the bounds on $\Fo$ or $\Fmo$ do not hold anymore. 

\section{Experimental Results}\label{sec:perf}

In this section, we evaluate all presented algorithms on random instances. We also introduce several other algorithms, for which we have no formal bounds, to understand them empirically. 
For most algorithms, it is extremely difficult to theoretically analyze the fraction of positive instances for a given load as we did for \greedyuniform in Sec.~\ref{sec:small}. By doing these experiments, we get an idea of the difference between the bounds on the load we have proved for worst-case instances and random instances. Moreover, it gives us insights on how well our algorithm perform when they are not in a load regime where they are guaranteed to find a solution.

Our instances are randomly generated, and not taken from a real dataset. The CRAN application we propose
is still at the prototype phase~\cite{guiraud2022experimental} and has not been used in the field, thus no data exists on the typical delays that we should consider. However, from actual telecommunication protocols and technologies, we know that the value of $P$ is at most $100.000$ and the number of messages is from a few tens to a few hundreds at most~\cite{bartharxiv2018deterministic}. All experiments are done within this range.

\subsection{Experimental Results for Large Messages} \label{sec:perf_large}

In this section, the performance on random instances of the algorithms presented in Sec.~\ref{sec:large} is experimentally characterized. The implementation in C of these algorithms can be found on github\footnote{\url{https://github.com/Mael-Guiraud/GuiraudStrozecki2023Scheduling}}. We experiment with several periods and message sizes. For each set of parameters, we compute the success rate of each algorithm for all possible loads by changing the number of messages. The success rate is measured on $10,000$ instances of \pma generated by drawing uniformly and independently the delays of each message in $[P]$.

We consider the following algorithms:
\begin{itemize}
  \item \firstfit
  \item \metaoffset
  \item \compactpair
  \item \compactfit
  \item \greedyuniform, the algorithm introduced and analyzed in Sec.~\ref{sec:small}, used for arbitrary $\tau$
  \item \texttt{Exact Resolution} which always finds a valid assignment if there is one, using an algorithm from~\cite{bartharxiv2018deterministic}  
\end{itemize}

\begin{figure}[h]
 \begin{center}
\includegraphics[scale=0.8]{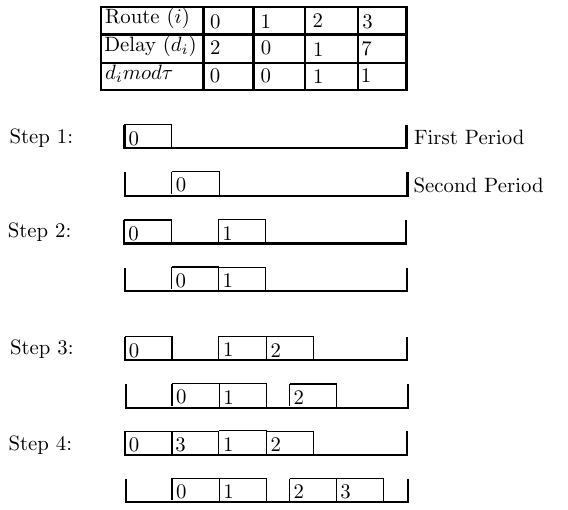}
\end{center}
\caption{A run of \compactfit with $\tau = 2$ and $P=10$, which creates two compact pairs}
\label{fig:compactfit}
\end{figure}

The only algorithm we have yet to describe is \compactfit. The idea is, as for \compactpair, to combine the absence of collision on the first period of \metaoffset and the compactness of assignments given by \firstfit.
The messages are sorted in increasing order of their delay modulo $\tau$, and each message is scheduled so that it extends an already scheduled compact tuple. 
In other words, it is scheduled using meta-offsets such that using one less as a meta-offset creates a collision on \emph{the second period}. If it is not possible to schedule the message in that way, the first possible meta-offset is chosen. See Fig.~\ref{fig:compactfit} for an example run of \compactfit. This algorithm is designed to work well on random instances. Indeed, it is not hard to evaluate the average size of the created compact tuples, and from that, to prove that \compactfit works with high probability when the load is strictly less than $1/2$.

On a regular $2017$ laptop, all algorithms terminate in less than a second when solving $10,000$ instances with $100$ messages except \texttt{Exact Resolution}, whose complexity is exponential in the number of messages (but polynomial in the other parameters). Hence, the exact value of the success rate given by \texttt{Exact Resolution} is only available in the experiment with at most $10$ messages (the algorithm cannot compute a solution in less than an hour for twenty messages and a high load).

\begin{center}
\includegraphics[scale=0.275]{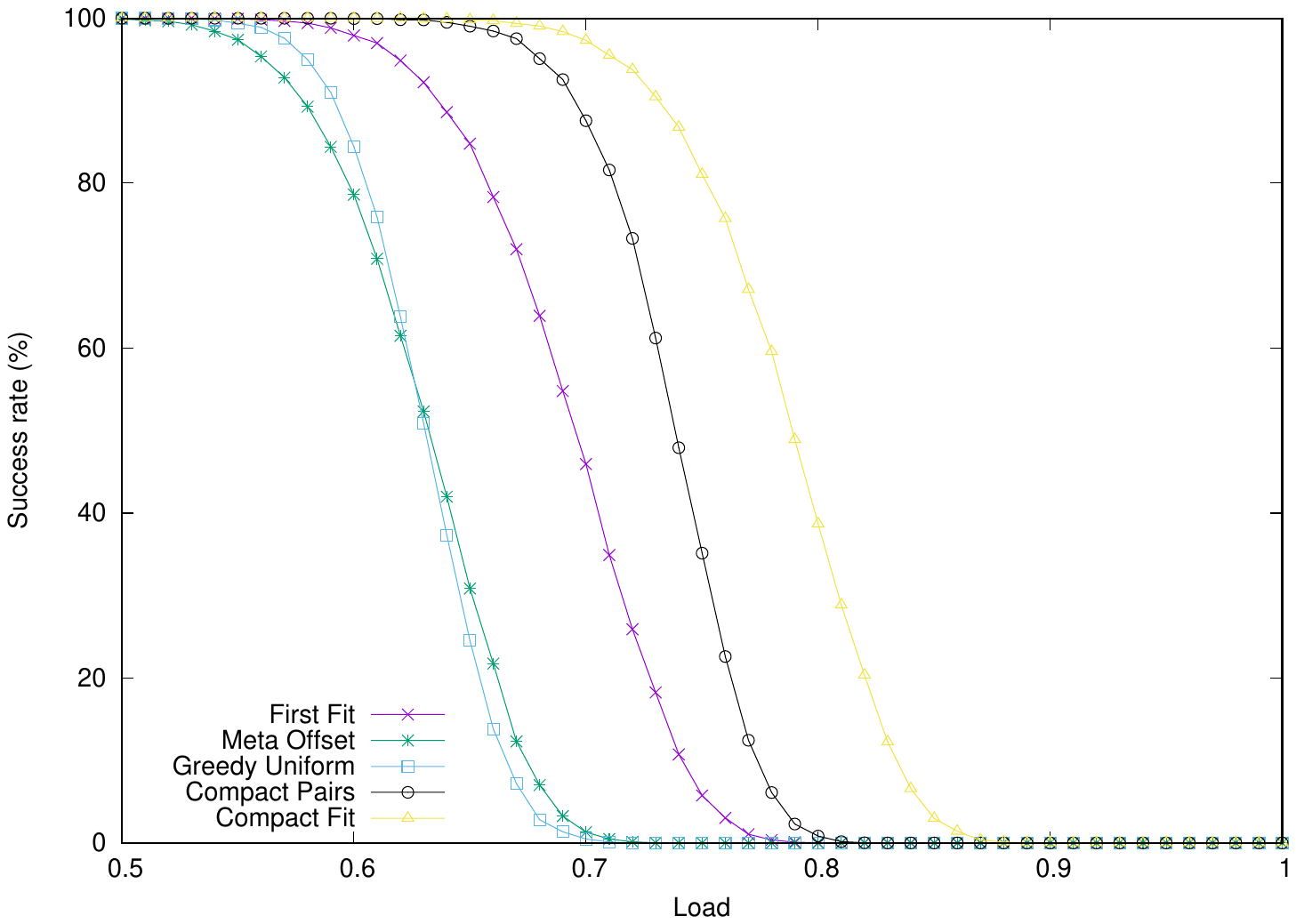}

\captionof{figure}{Success rates of all algorithms for increasing loads, $\tau = 1000$, $P=100,000$}
\label{fig:100messBig}
\end{center} 

\begin{center}  
\includegraphics[scale=0.275]{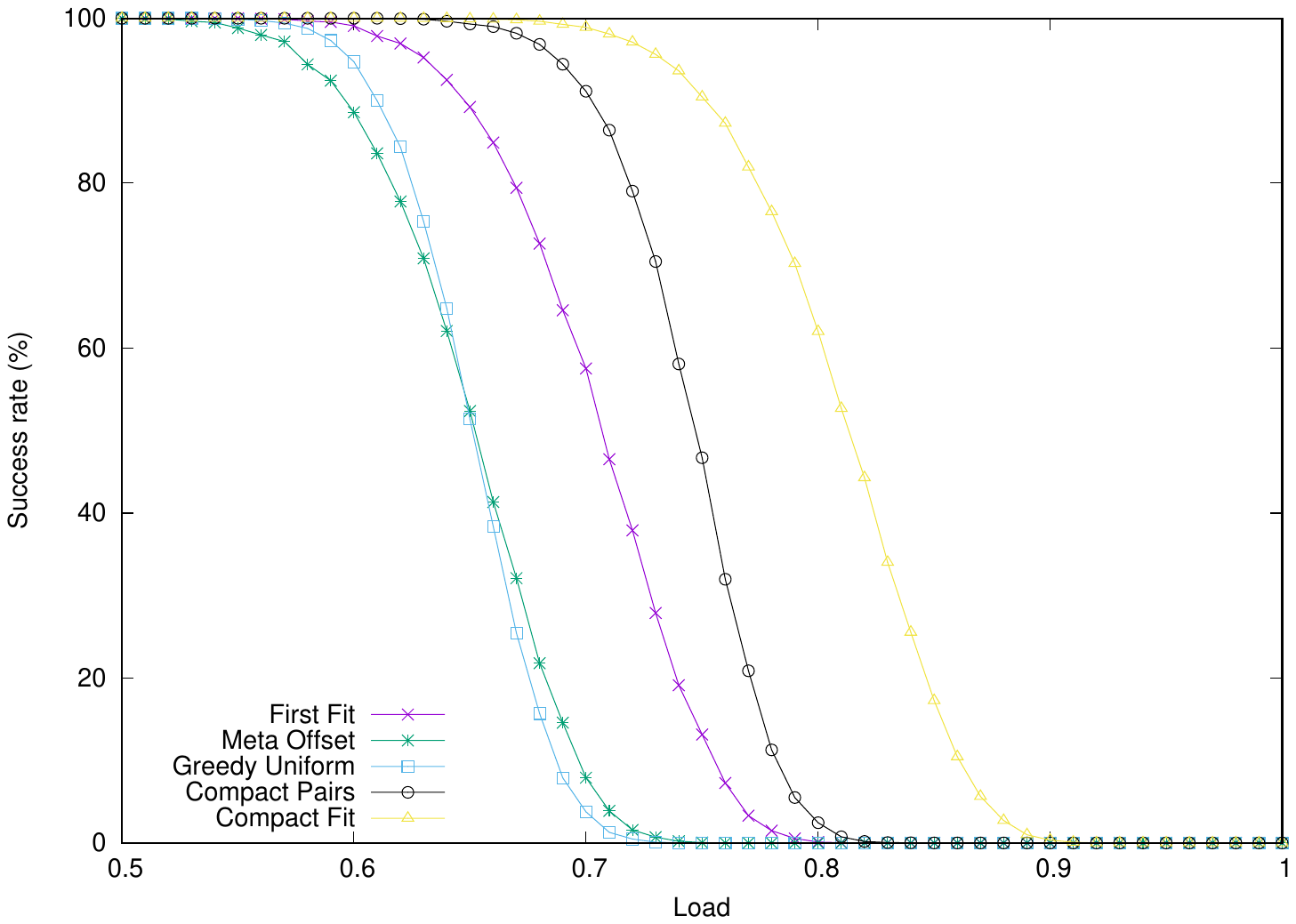}
\captionof{figure}{Success rates of all algorithms for increasing loads, $\tau = 10$, $P=1000$}
\label{fig:100messSmall}
\end{center}

\begin{center}
\includegraphics[scale=0.275]{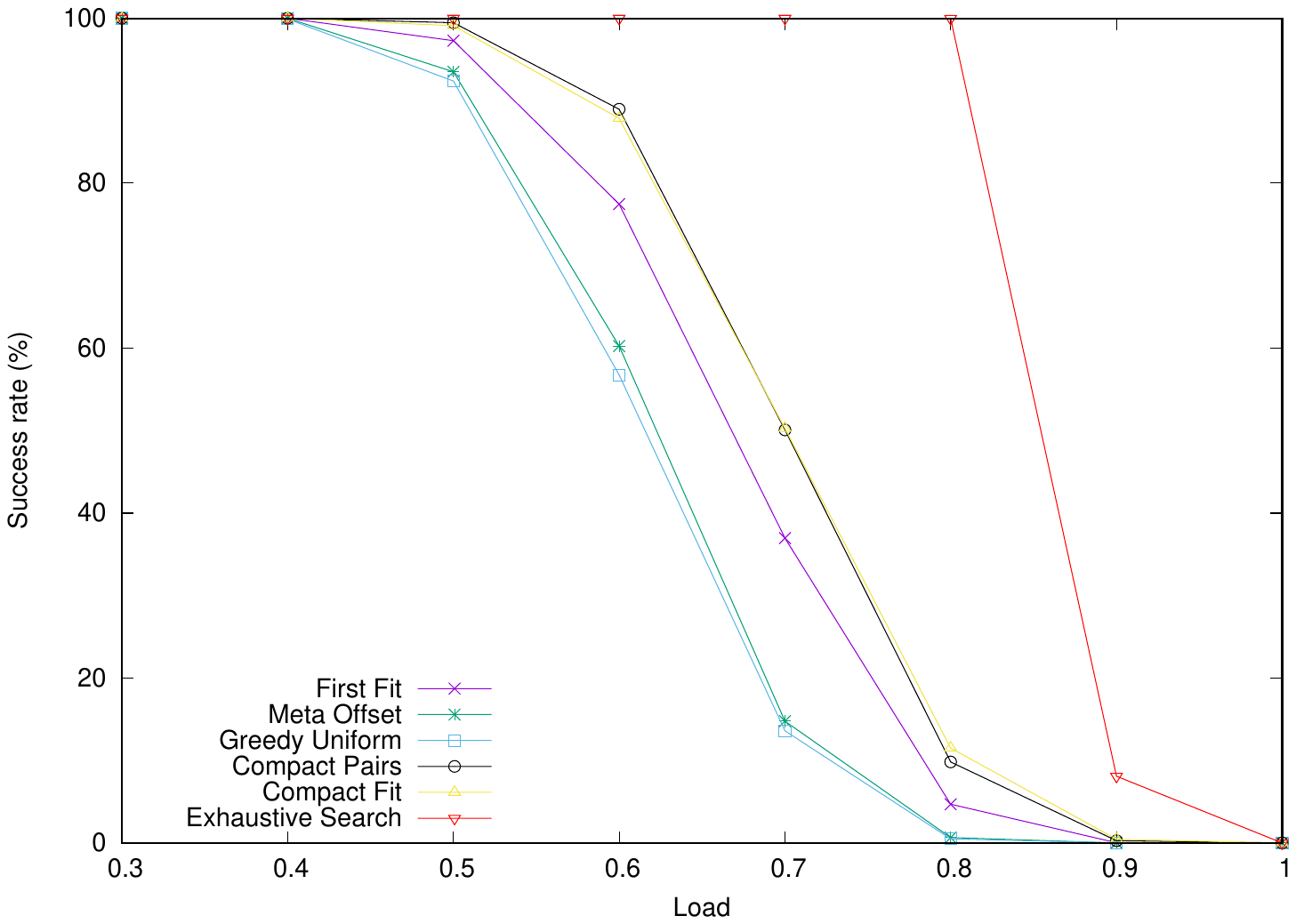}
\end{center}
\captionof{figure}{Success rates of all algorithms for increasing loads, $\tau = 1000$, $P=10,000$}
\label{fig:10mess}

\begin{center}
\includegraphics[scale=0.275]{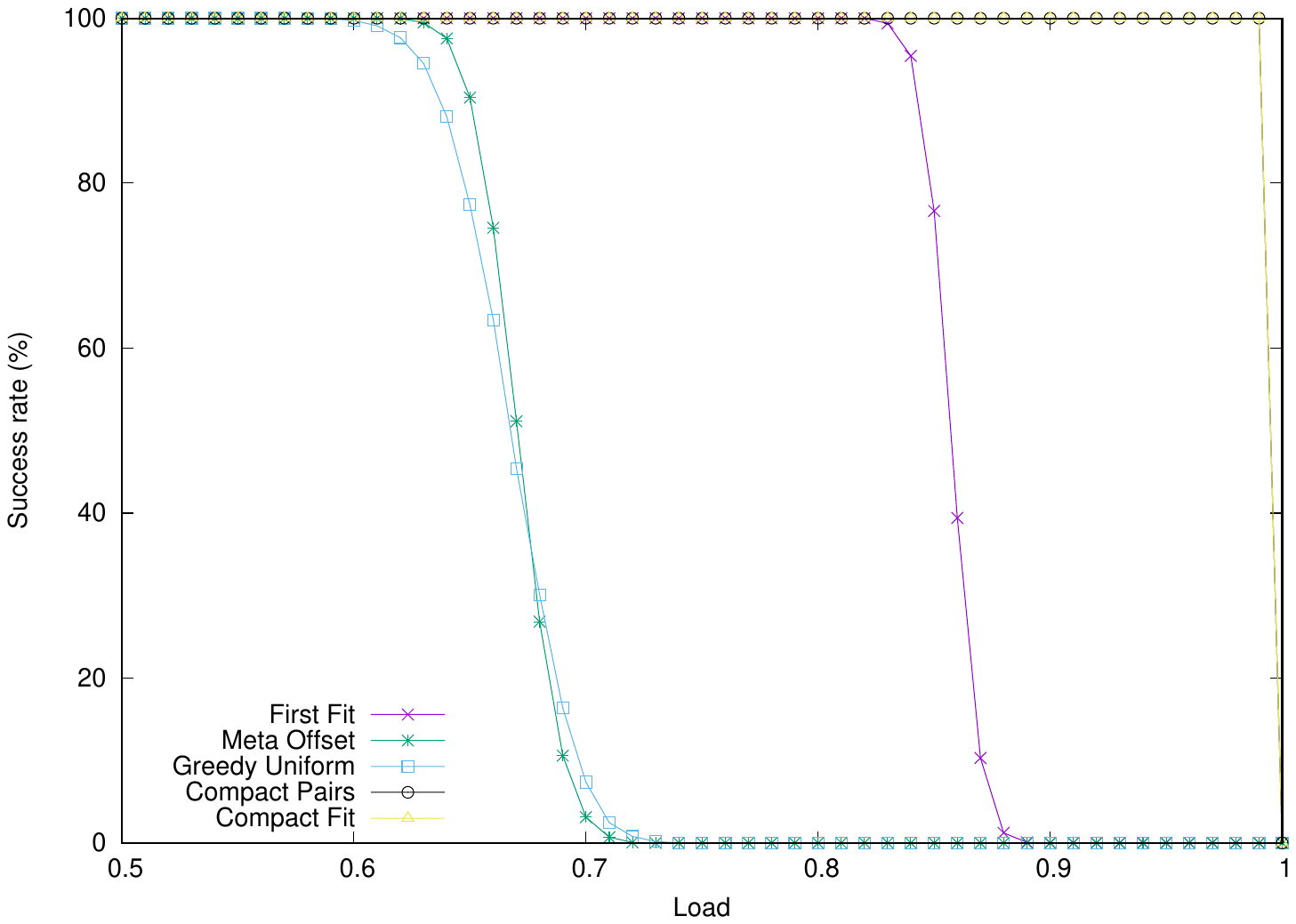}
\end{center}
\captionof{figure}{Same parameters as in Fig.~\ref{fig:100messBig}, delays uniformly drawn in $[\tau]$}
\label{fig:shortroutes}

\medskip

For all sets of parameters, the algorithms have the same relative performances. \metaoffset and \greedyuniform perform the worst and have almost equal success rates. Remark that they have a $100\%$ success rate for load less than $1/2$, while it is easy to build an instance of \pma of load $1/3 +\varepsilon$ which makes them fail. 

\firstfit performs better than \metaoffset on random instances, while we have proved that
they always find a valid assignment for load at most $1/3$ but not above.
 \compactpair, for which we have proved a better bound on the load also performs well in the experiments, always finding assignments for a load of $0.6$. \compactfit is similar in spirit to \compactpair but is designed to have a good success rate on random instances is indeed better than \compactpair, when there are enough messages.

As demonstrated by Fig.~\ref{fig:100messBig} and Fig.~\ref{fig:100messSmall}, the size of the messages has little impact on the success rate of the algorithms, when the number of messages and the load are kept the same. Comparing Fig.~\ref{fig:10mess} and Fig.~\ref{fig:100messBig} shows that for more messages, the transition between $100\%$ success rate to $0\%$ success rate happens faster.
Finally, the results of \texttt{Exact Resolution} in Fig.~\ref{fig:10mess} show that the greedy algorithms are far from always finding a solution when it exists. Moreover, we have found an instance with load $0.8$ with no assignment found by \texttt{Exact Resolution}, which gives an \textbf{upper bound on the load} for which \pma can always be solved positively.

We also investigate the behavior of the algorithms when the delay of messages is drawn in $[\tau]$ in Fig.~\ref{fig:shortroutes}. The difference from the case of large delay is that \compactpair and \compactfit are extremely efficient: they always find a solution for $99$ messages. It is expected since all $d'_i$ are equal in these settings, and they will both build a $99$-compact tuple and thus can only fail for load $1$.

\subsection{Experimental Results for Small Messages} \label{sec:perf_small}

In this section, the performance on random instances of the algorithms presented in Sec.~\ref{sec:small} is experimentally characterized. The settings are as in Sec.~\ref{sec:perf_large}, with $\tau = 1$. The evaluated algorithms are:

\begin{itemize}
  \item \firstfit
  \item \greedyuniform 
  \item \greedypotential, a greedy algorithm which leverages the notion of potential introduced for Swap. 
  It schedules messages in arbitrary order, choosing the available offset which maximizes the potential of the unscheduled messages
  \item \swapandmove 
  \item \texttt{Exact Resolution}
\end{itemize}

As in Sec.~\ref{sec:perf_large}, the success rate on random instances is much better than the bound given by the worst-case analysis of the article. In the experiment presented in Fig.~\ref{fig:tau1}, all algorithms succeed on all instances when the load is less than $0.64$. \greedyuniform behaves exactly as proved in Th.~\ref{theorem:uniform}, with a very small variance. The performance of \swapandmove and its simpler variant \greedypotential, which optimizes the potential in a greedy way, is much better than \firstfit or \greedyuniform. Amazingly, \swapandmove always finds an assignment when the load is less than $0.95$. \swapandmove is extremely close to Exact Resolution, but for $P=10$ and load $0.9$ or $1$, it fails to find some valid assignments, as shown in Fig.~\ref{fig:tau1-10mess}.

\begin{center} 
\includegraphics[scale=0.275]{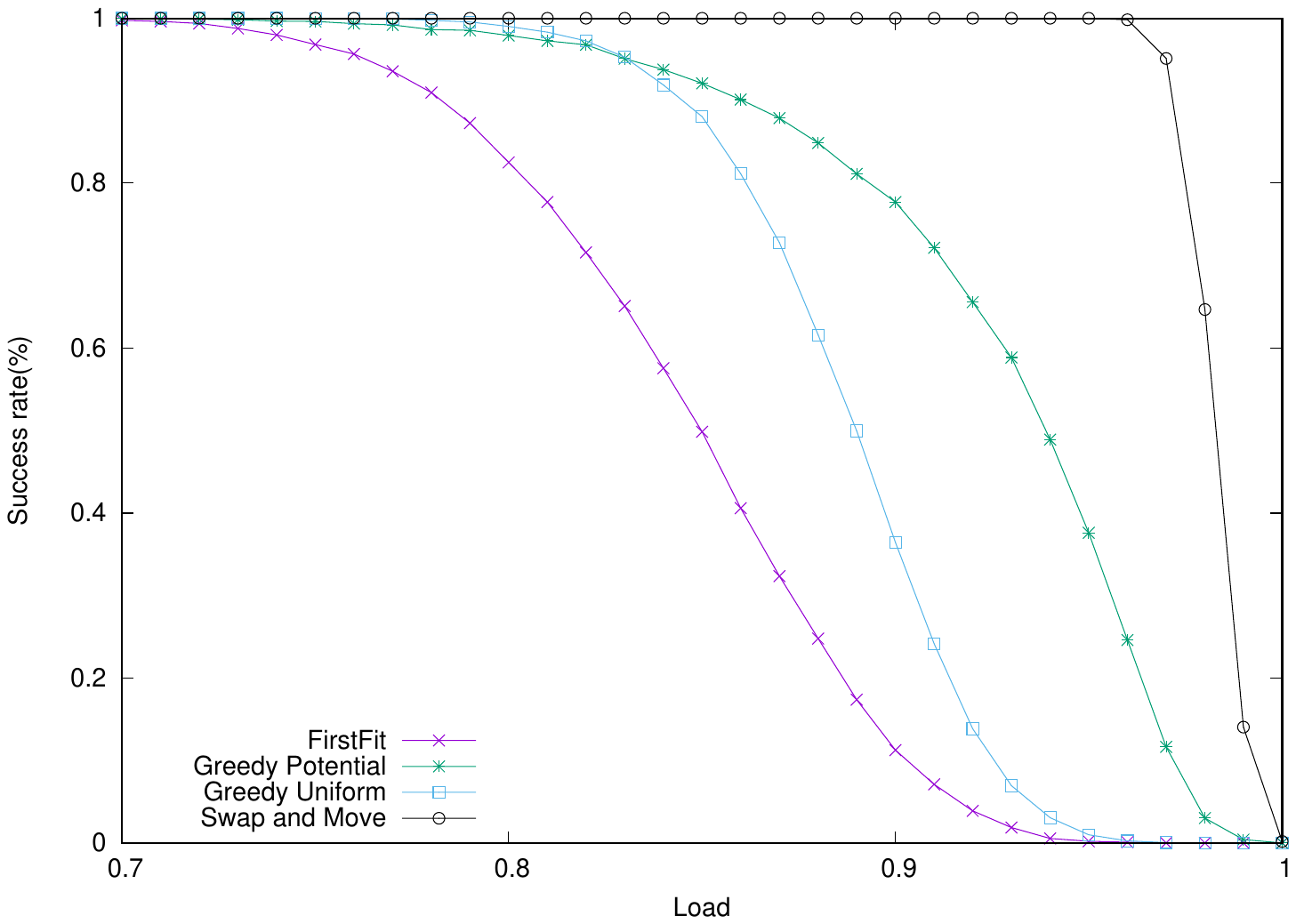} 
\end{center}
\captionof{figure}{Success rates of all algorithms for increasing loads, $\tau = 1$ and $P=100$}
\label{fig:tau1}

\begin{center}
\includegraphics[scale=0.275]{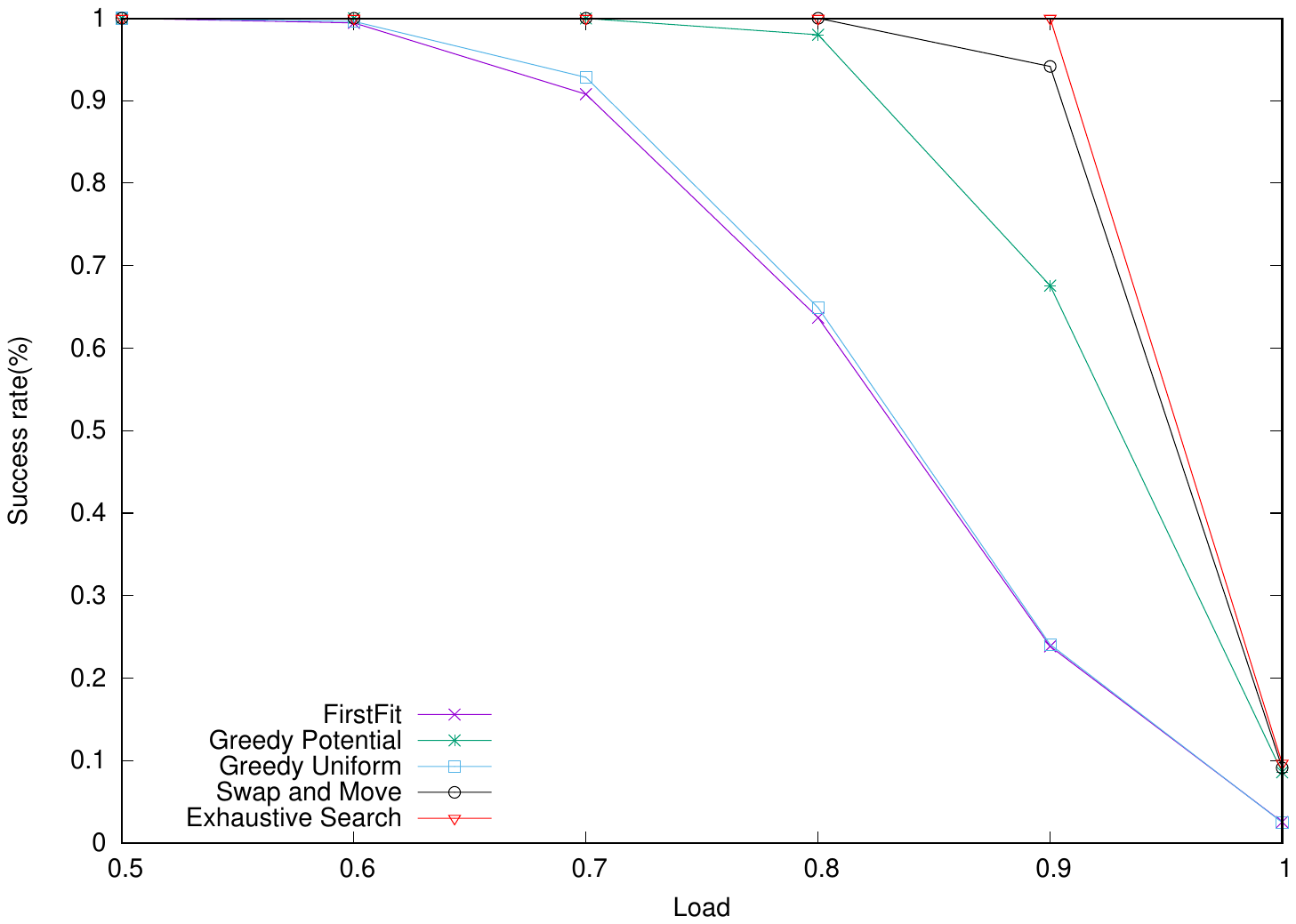}
\end{center}
\captionof{figure}{Success rates of all algorithms for increasing loads, $\tau = 1$ and $P=10$}
\label{fig:tau1-10mess}

 \medskip
 Finally, we evaluate the computation times of the algorithms to understand whether they scale to large instances. We present the computation times in Fig.~\ref{fig:timelog} and we choose to consider instances of load $1$, since they require the most computation time for a given size. The empirical complexity of an algorithm is evaluated by
 linear regression on the function that associates to $\log(n)$, the log of the computation time of the algorithm on $n$ messages.  \firstfit, \greedyuniform, and \swapandmove scale almost in the same way, with an empirical complexity slightly below $O(n^2)$, while \greedypotential has an empirical complexity of $O(n^3)$. The empirical complexity corresponds to the worst-case complexity we have proved, except for \swapandmove which is in $O(n^3)$. There are two explanations for this difference between average case complexity and worst case complexity: most of the messages are scheduled by the fast \firstfit subroutine and most Swap operations improve the potential by more than $1$, as we assume in the worst-case analysis.

\begin{figure}
 \begin{center}
\includegraphics[scale=0.275]{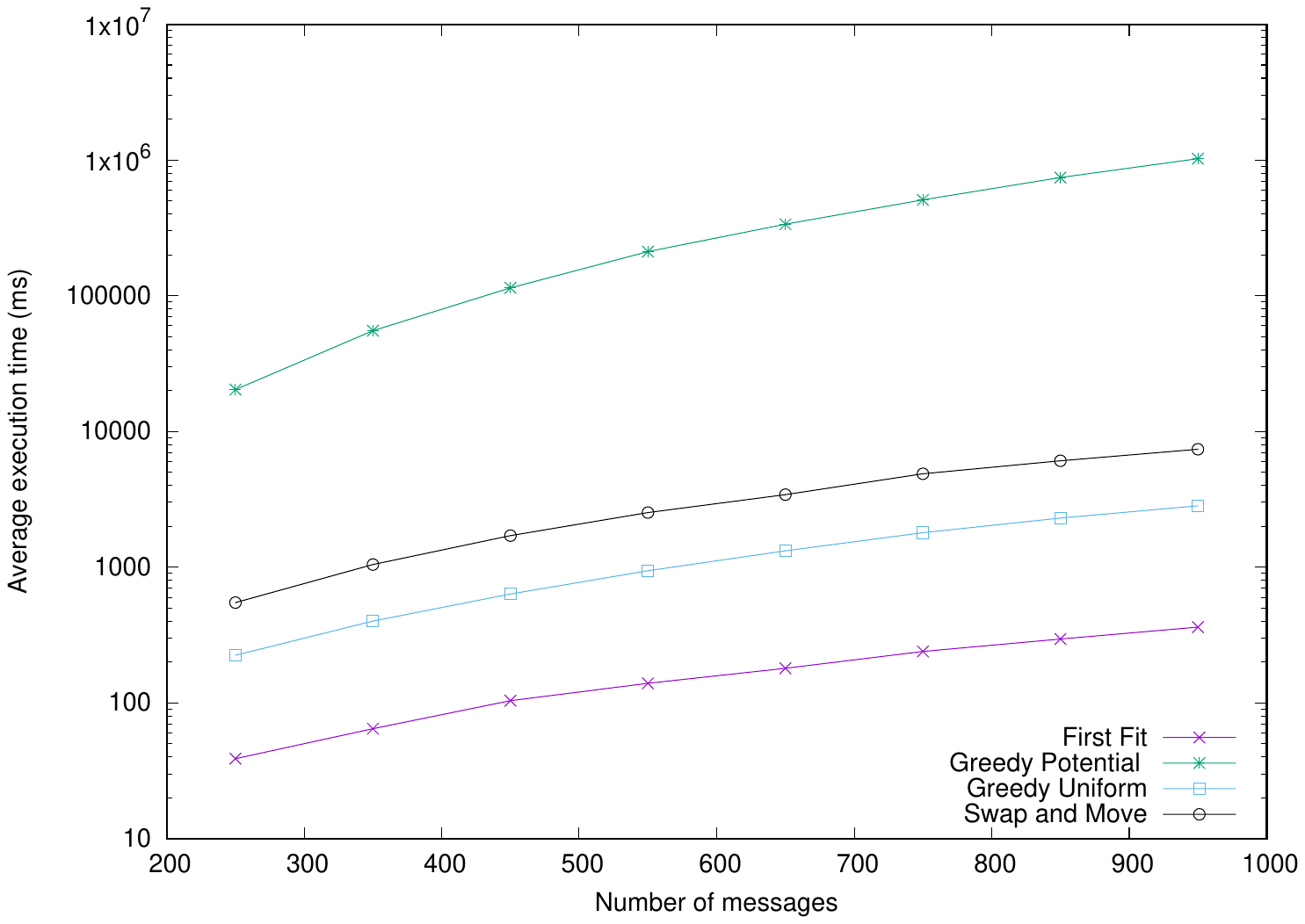}
\end{center}
\caption{Computation time (logarithmic scale) function of the number of messages of all algorithms on $10,000$ instances of load $1$}
\label{fig:timelog}
\end{figure}

\section{Conclusion}

In this article, we have proved that there is always a solution to \pma and that it can be found by a polynomial time greedy algorithm for arbitrary message size and load at most $2/5$. For messages of size $1$ and load at most $\phi - 1$, a solution is found by a polynomial time local search algorithm. Moreover, the presented algorithms find valid assignments for random instances for much higher loads as we have shown empirically but also theoretically for a randomized greedy algorithm. As a consequence, we obtain communication schemes for C-RAN with \emph{no buffering nor logical latency}, even for quite loaded fronthaul networks.

The first limitation of our model is the topology of the network. We have shown that an arbitrary topology can be transformed into a topology with two contention points. This transformation, while preserving the load, may map positive instances into negative instances. Hence, it would be interesting to design algorithms working directly on complex topologies, as we have done for the periodic assignment problem with buffering in~\cite{guiraud2021deterministic}. 

The second limitation is that all messages are of the same size. 
To model networks with different kinds of traffic or a production line with different tasks, we should relax this hypothesis. An interesting direction of research would be to adapt algorithm \swapandmove to this setting. As a first step, we could already adapt \swapandmove to a single arbitrary $\tau$ and try to improve on the bound on latency obtained using \texttt{Compact 8-tuples}.

Finally, it remains to prove that \pma is \NP-complete. If it is indeed the case, it would be interesting to understand for which load \pma can always be positively solved by exact algorithms. 

\bibliography{Sources}

\end{document}